\documentclass[12pt]{article}
\usepackage{stmaryrd}
\usepackage{amsmath}
\usepackage{graphicx}
\usepackage{color}
\usepackage[title]{appendix}
\usepackage{amsfonts,amscd,amssymb, mathtools,mathrsfs, dsfont}%, bbm,bbding} %math and other symbols
\usepackage{pgfplots}
\usepackage{float}
\usepackage{url}
\usepackage{hyperref}
\usepackage{array}
\hypersetup{
colorlinks=true, %colorise les liens
breaklinks=true, %permet le retour ??? la ligne dans les liens trop longs
urlcolor=blue, %couleur des hyperliens
linkcolor=red, %couleur des liens internes
citecolor=purple, %couleur des r???f???rences
pdftitle={}, %informations apparaissant dans
pdfauthor={}, %les informations du document
pdfsubject={} %sous Acrobat.
}
\tikzset{elegant/.style={smooth,thick,samples=50,cyan}}
\tikzset{eaxis/.style={->,>=stealth}}
\newif\ifblog
\newif\iftex
\blogfalse
\textrue

\newcommand{\thmref}[1]{Theorem~{\rm \ref{#1}}}
\newcommand{\lemref}[1]{Lemma~{\rm \ref{#1}}}

\def\wv{\widehat{v}}
\def\WV{\widehat{V}}
\def\WH{\widehat{H}}
\def\WI{\widehat{I}}

\def\wu{\widehat{u}}

\def\d{{\rm d}}
\def\P{{\mathbb P}}
\def\E{{\mathbb E}}

\def\Q{{\cal Q}}

\def\p{{\partial}}

\newcommand{\ep}{\varepsilon}
\newcommand{\al}{\alpha}

\newcommand{\nd}{\noindent}

\newcommand{\la}{\lambda}
\newcommand{\La}{\Lambda}
\newcommand{\ga}{\gamma}

\newcommand{\de}{\delta}
\newcommand{\Dt}{\Delta t}
\newcommand{\be}{\begin{eqnarray}}
\newcommand{\ee}{\end{eqnarray}}
\newcommand{\bee}{\begin{eqnarray*}}
\newcommand{\eee}{\end{eqnarray*}}
\newcommand{\rt}{\rightarrow}
\newcommand{\argmax}{\mathop{\rm argmax}\limits}

\newtheorem{theorem}{Theorem}[section]
\newtheorem{lemma}[theorem]{Lemma}

\newtheorem{remark}[theorem]{Remark}

\newenvironment{proof}{\noindent {\sc Proof:}}{\strut\hfill $\Box$} %\medskip}

\newcommand{\LL}{{\cal L}}

\renewcommand{\geq}{\geqslant}
\renewcommand{\leq}{\leqslant}

\newcommand{\calN}{\mathbb{N}}
\newcommand{\cI}{\mathbb{I}}
\newcommand{\cL}{\mathbb{L}}
\newcommand{\cH}{\mathbb{H}}
\newcommand{\aH}{{\cal H}}
\newcommand{\QLT}{{\cal Q}_{L,T}}
\newcommand{\QT}{{\cal Q}_{T}}
\newcommand{\TQ}{{\cal Q}^{T}}
\newcommand{\ND}{{\cal ND}}
\newcommand{\cD}{{\cal D}}

\newcommand{\PR}{{\cal R}}
\newcommand{\NR}{{\cal NR}}
\newcommand{\NA}{{\cal NA}}
\newcommand{\hatz}{\hat{z}}
\newcommand{\esssup}{\mathrm{ess\:sup\:}}
\newcommand{\vt}{\tau}%{\vec{t}\:}
\newcommand{\vs}{T-s}
\newcommand{\deathtime}{\theta}

\title{Dynamic optimal reinsurance and dividend-payout\\ in a finite time horizon}

\author{Chonghu Guan\thanks{School of Mathematics, Jiaying University, Meizhou 514015, Guangdong,
China. This author is partially supported by NSF of Guangdong Province of China (No. 2021A1515012031), NNSF of China (No. 11901244), and Guangdong Basic and Applied Basic Research Foundation (No. 2022A1515010263).
 Email: \url{316346917@qq.com}.} \and Zuo Quan Xu\thanks{Department of Applied Mathematics, The Hong Kong Polytechnic University, Kowloon, Hong Kong, China. Partially supported by NSFC (No.11971409), The Hong Kong RGC (GRF No.15202421 and No.15202817), The PolyU-SDU Joint Research Center on Financial Mathematics, The CAS AMSS-PolyU Joint Laboratory of Applied Mathematics, The Hong Kong Polytechnic University. Email: \url{maxu@polyu.edu.hk}.} \and Rui Zhou\thanks{Department of Applied Mathematics, The Hong Kong Polytechnic University, Hong Kong, China. Partially supported by NSFC (No.11971409), The Hong Kong RGC (GRF No.15202421 and No.15202817), The PolyU-SDU Joint Research Center on Financial Mathematics, The CAS AMSS-PolyU Joint Laboratory of Applied Mathematics, The Hong Kong Polytechnic University. Email: \url{rui.zhou@connect.polyu.hk}.
}}
\date{}

\begin{document}
\maketitle

\begin{abstract}
This paper studies a dynamic optimal reinsurance and dividend-payout problem for an insurance company in a finite time horizon. The goal of the company is to maximize the expected cumulative discounted dividend payouts until bankruptcy or maturity which comes earlier. The company is allowed to buy reinsurance contracts dynamically over the whole time horizon to cede its risk exposure with other reinsurance companies. This is a mixed singular-classical control problem and the corresponding Hamilton-Jacobi-Bellman equation is a variational inequality with a fully nonlinear operator and subject to a gradient constraint. We obtain the $C^{2,1}$ smoothness of the value function and a comparison principle for its gradient function by the penalty approximation method so that one can establish an efficient numerical scheme to compute the value function. We find that the surplus-time space can be divided into three non-overlapping regions by a risk-magnitude and time-dependent reinsurance barrier and a time-dependent dividend-payout barrier. The insurance company should be exposed to a higher risk as its surplus increases; be exposed to the entire risk once its surplus upward crosses the reinsurance barrier; and pay out all its reserves exceeding the dividend-payout barrier. The estimated localities of these regions are also provided.
\bigskip

\nd {\bf Keywords.} Optimal reinsurance; optimal dividend; free boundary problem; stochastic optimal control; mixed stochastic control.

\bigskip
\nd {\bf Mathematics Subject Classification.}
%91B60.%Trade models
91G05; %Actuarial mathematics
35R35; %Free boundary problems
91B70; %Stochastic models

\end{abstract}

\setlength{\baselineskip}{0.25in}

\section{Introduction}
A fundamental goal of an insurance company (also called insurer or ceding company) is to improve its solvency and stability. This goal can be reached by choosing a good dividend-payout and reinsurance strategy, that is, determining how many of the insurance company's reserves to be paid out as dividends to its shareholders and the proportion of its risk exposure to be shared with other reinsurance companies (also called reinsurers or ceded companies). The first strategy requires a trade-off between the insurance company and its shareholders. Shareholders would react positively when the dividend payouts increase, however, paying out dividends would also reduce the future reserves that are essential for the survival of the company. Due to the solvency capital requirements for large losses in emergency incidents, such as the ``Black Swan" event, the outbreak of COVID-19, the insurance company tends to cede part of its risk exposure to other reinsurance companies at the expense of reinsurance premiums. So the second strategy comes in and it requires a trade-off between the insurance and reinsurance companies. The reinsurance companies promise to cover certain part of the insurance company's risk exposure, increasing the insurance company's survival chance; meanwhile, the insurance company has to pay reinsurance premiums to the reinsurance companies, decreasing the survival chance. This paper considers a single model that takes both the dividend-payout and reinsurance strategies into consideration. Our objective is to find the optimal strategies for the insurance company so as to maximize its expected cumulative discounted dividend payouts until its bankruptcy or a given maturity time which comes earlier.

The optimal dividend-payout problem and reinsurance is well studied in the mathematical insurance literature; see, e.g., \cite{TZ98, Ta00, WY10}. As the dynamic surplus process must be stochastic and depend on the insurance company's strategies, the problem usually boils down to a stochastic control problem. Two types of objectives for the insurance company are widely considered: one is to minimize the ruin probability, e.g., \cite{DW96, Ce97, Sc04}, and the other is to maximize dividend payouts until bankruptcy, e.g., \cite{AT09, CGY06, JP12, LLWX20}. Many existing studies on this topic show that it is optimal to pay out dividends according to a band strategy; see Azcue and Muler \cite{PN14} for a complete exposition on the topic. For instance, Taksar \cite{Ta00} considered a model in which the surplus of an insurance company follows a controlled diffusion process and the insurance company pays out dividends in two different cases. In the first case, the dividend-payout rate is constrained to a bounded interval $[0,l]$. The optimal dividend-payout strategy turns out to be an ``all or nothing" policy with respect to (w.r.t., for short) the dividend rate. That is, paying out dividends at the minimum rate 0 if the surplus is lower than a threshold, and paying out at the maximum rate $\ell$ otherwise. In the second case, the dividend-payout rate is unbounded. The optimal strategy is to keep the surplus under some barrier, that is, paying out all reserves exceeding the barrier as dividends, and doing nothing under the barrier. The optimal reinsurance strategy, which clearly depends on the surplus level, suggests that the insurance company should be exposed to a higher risk as its surplus increases. Along this line of research, other factors including, but not limited to, liability, regime switching, have also been taken into account in model formulation; see \cite{TZ98,WY10}.

When study an optimal dividend problem, people often transform it into an optimal stopping problem where the value function is the derivative of the original problem's value function.
For instance, De Angelis and Ekstrom \cite{AE17} linked a finite time optimal dividend problem to an optimal stopping problem for a diffusion reflected at 0. The spatial and time regularities of the value function of the optimal stopping problem are established, so that they can show the value function of the original optimal dividend problem is the unique classical solution of a suitable Hamilton-Jacobi-Bellman equation. This work was extended to the partial information case in \cite{A15}. Following the ideas developed in \cite{A15,AE17}, Bandini, et al., \cite{BAFG20} introduced the stochastic discounting into the model, leading to a two-dimensional spatial value function. Ferrari and Schuhmann \cite{FS19} linked an optimal dividend problem with capital injections to an optimal stopping problem for a drifted Brownian motion that is absorbed at the origin. They showed that whenever the optimal stopping rule is triggered by a time-dependent boundary, the value function of the optimal stopping problem gives the gradient function (i.e., the derivative function of the value function) of the optimal dividend problem. The aforementioned paper mainly used probabilistic arguments.

Insurance market has grown furious in recent years. Insurance companies, especially the global insurance companies, such as AIA Group Ltd., AIG, China Life Insurance Company, operate diverse businesses and offer insurance products on various term basis, such as property and casualty insurance, life insurance, health insurance. In this regard, it is natural for insurance companies to consider their total risk exposure as a combination of different types of risks. Mathematically speaking, the total risk exposure may take any, discrete or continuous, probability distributions. Similarly, but more importantly, the reinsurance strategies should also be based on different types of risks, resulting in non-classical complicated reinsurance policies. Although the optimal reinsurance problem has been extensively investigated, most of the existing literature analyzes only for typical reinsurance policies such as proportional and excess of loss reinsurance; see, e.g., \cite{Ta00, SH01,TM03,HV03, HHL08, HT10, LLWX20}. This paper investigates an optimal reinsurance and dividend-payout problem, which brings more practical features. We do not restrict ourself to these typical reinsurance policies. We consider a controlled diffusion surplus process, which is a good approximation of the classical Cram\'{e}r-Lundberg process as well-justified by Grandell \cite{Gr12}. A closest model to this paper is considered by Guan, Yi and Chen \cite{GY19(2)}, where the risk control model is relatively simple and the type of reinsurance policy is constrained to the proportional one. It turns out that the reinsurance scheme is restrictive. Tan, et al., \cite{TWWZ20} considered a similar infinite time horizon problem, so the corresponding Hamilton-Jacobi-Bellman (HJB) equation is an ordinary differential equation (ODE), which is much easier to handle than the partial differential equation (PDE) as in our case.

In our model, both the drift and volatility of the controlled surplus process depend on the reinsurance policy. The target of the insurance company is to maximize the expected cumulative discounted dividend payouts until bankruptcy or a given maturity which comes earlier. The problem is a mixed singular-classical stochastic control problem.
The model has the following features: First, we do not confine the reinsurance contracts to be some particular ones such as proportional reinsurance or excess of loss reinsurance. The insurance company can freely choose its reinsurance policy of any type subject to the expected value premium principle. Second, the reinsurance contracts are chosen dynamically depending on the surplus level. It turns out that the optimal reinsurance policy is a feedback one that depends on the surplus level and time. Third, the insurance claims can admit any probability distribution subject to a tail constraint. Especially for bounded claims, we have a better understanding of the optimal reinsurance strategy. Finally, we consider a finite time horizon problem leading to an extremely challenging fully nonlinear HJB equation subject to a gradient constraint.
\par 
The PDE/ODE method (such as the viscosity solution method) is also often adopted to study singular and classical control problems, especially for finite time horizon problems. Compared to the probabilistic method, one of the main advantages of this method is that one can freely transform the value function to another PDE that is easier to study. The new PDE does not necessarily have a stochastic control background. For instance, \cite{G13} made a standard exponential transformation to get a backward heat equation in a symmetrical region. But his method cannot guarantee the smoothness of the function, so the author considered approximated problems. The value functions of the approximated problems are guaranteed to be smooth so that they can be computed by the principle of smooth fit and Green’s function methodology. Eventually $\ep$-optimal strategies for the original problem are constructed. The author further studied the behavior of the free boundary in \cite{G14} and extended the model subject to a ruin probability constraint in \cite{G15}. 
\par 
We now highlight some mathematical contributions of this paper. Same as many existing optimal dividend models, the HJB equation of our problem turns out be a variational inequality problem subject to a gradient constraint. However, since the reinsurance policies can be time-dependently chosen, there is a functional optimization problem appearing in the operator of the HJB equation, which makes the operator a fully nonlinear one. By contrast, De Angelis and Ekstrom \cite{AE17} considered a model without such reinsurance policies, so the operator is linear, and consequently, the gradient function can be naturally linked to an optimal stopping problem for a diffusion process. In our case, the operator is fully nonlinear so that we cannot find such an underlying process. This prevents us to link the gradient function to some optimal stopping problem, so the ideas developed in \cite{A15} and \cite{AE17} cannot be applied to our problem. To overcome this difficulty, we adopt a pure PDE approach to study the HJB equation. Inspired by \cite{DY09a, DY09b, DXZ10}, we find out that the gradient function satisfies an obstacle PDE. The latter is then studied by the penalty approximation method. Different from \cite{G13, G14, G15}, we derive a comparison principle and establish necessary properties (such as $C^{2,1}$ smoothness and uniqueness of the value function) for the original full nonlinear HJB equation. As a byproduct, one can compute the value function as well as the optimal strategies for our problem by establishing an efficient numerical scheme to solve the approximation PDE. Our approach requires a lot of deep results in PDE and functional analysis, such as the Leray-Schauder fixed point theorem, the Sobolev embedding theorem, the $C^{\al,\frac{\al}{2}}$ estimation, the Schauder estimation, and the comparison principle for nonlinear PDEs.

\par
Our model provides a lot of economic insights as well. We show that there is a smooth, time-dependent, dividend-payout barrier that divides the surplus-time space into a no-dividend-payout region and a dividend-payout region. The insurance company should pay out all its reserves exceeding the dividend-payout barrier (that is, all the reserves in the dividend-payout region). Furthermore, we find a risk-magnitude and time-dependent smooth reinsurance barrier that divides the no-dividend-payout region into a reinsurance-covered region and a reinsurance-uncovered region, in an increasing order of the surplus. Therefore the reinsurance is an excess of loss reinsurance. Economically speaking, when faced with the same magnitude of risk, the insurance companies with higher surpluses tend not to cede their risks to other reinsurance companies. In other words, as the magnitude of the risk is getting smaller, the reinsurance-covered region is shrinking (i.e., less insurance companies with different surpluses tend to cover the risks of this magnitude); whereas, the reinsurance-uncovered region is expanding. The former region disappears when the magnitude of the risk is smaller than an explicitly given constant (namely, all the insurance companies choose to cover their entire risks by themselves); by contrast, the latter never vanishes. The insurance company should be exposed to a higher risk as its surplus increases in the reinsurance-covered region; be exposed to the entire risk once its surplus falls into the reinsurance-uncovered region. We also provide accurate estimations for the localities of these regions. Particularly, when the claims have a bounded distribution, we show that there is a uniform non-action region, in which the insurance company should be exposed to the entire risk and not pay out dividends.

The reminder of the paper is organized as follows. In Section 2, we formulate an optimal reinsurance and dividend-payout problem. Section 3 poses the corresponding HJB equation and gives a verification theorem. The properties such as the existence and uniqueness of a classical solution to the HJB problem are also provided. Section 4 is devoted to the study of the reinsurance, reinsurance-uncovered region and dividend-payout region. The proofs for our main results are given in the appendices.

%%%%%%%%%%%%%%%%%%%%%%%%%%%%%%%%%%%%%%%%%%%%%%

\section{Model Formulation}\label{sec:model}
\setcounter{equation}{0}

This paper investigates an optimal reinsurance and dividend-payout problem for an insurance company. We first need to model the cash reserves of the company, let us start from the classical Cram\'er-Lundberg model.

In the classical Cram\'er-Lundberg model, there are two components that affect the cash reserve (also called surplus) dynamics. The first one is the receiving payments of premiums from the policyholders at a constant rate $p$ continuously. The other is the outgoing payments for insurance claims. If we denote the total number of claims received until time $t$ by $\calN_t$ and the size of the $i$th claim by $Z_i$, then $R_t$, the company's surplus at time $t$, is given by
\begin{equation}\label{Rt1}
R_t=R_0+p t-\sum_{i=1}^{\calN_t}Z_i,
\end{equation}
where $\{\calN_t\}_{t\geq 0}$ is a Poisson process with intensity $1$, all the claims $Z_i$, $i=1,2,\cdots,$ are independent and identically distributed random variables, and they are independent of the Poisson process.

Sometimes, the insurance company needs to protect itself by sharing its risk exposure with other reinsurance companies. The insurance company buys reinsurance contracts from reinsurance companies dynamically. Given a reinsurance contract $I(\cdot)$, the insurance company can get a compensate $I(z)$ when the real loss magnitude is $z$. This function $I(\cdot)$ is known as the ceded loss function, and $H(z)\equiv z-I(z)$ is known as the retained loss function. The insurance company's reinsurance policy (or strategy) consists of purchasing a series of reinsurance contracts $\{I_{t}\}_{t\geq 0}$ over time, where $I_t$ denotes the reinsurance contract purchased at time $t$. Note that the reinsurance contracts are dynamically purchased by the insurance company, so they are usually time and surplus dependent.

The presence of reinsurance contracts modifies the risk exposure of the insurance company. It distorts the incoming and outgoing cash flow of the insurance company's surplus process \eqref{Rt1}. As well-justified by Grandell \cite{Gr12}, the surplus process $R_t$ can be approximated by the following diffusion process
\begin{equation}\label{Rt3}
\d R_t=\big(p-p(I_t)-\E_{R_{t-}}[Z_1-I_t(Z_1)]\big)\d t+\sqrt{\E_{R_{t-}}[(Z_1-I_t(Z_1))^2]}\;\d W_t,
\end{equation}
where $\{W_t\}_{t\geq0}$ is a standard Brownian motion independent of the random claims $Z_{i}$, $\E_{R_{t-}}[\;\cdot\;]=\E[\;\cdot\;|\;R_{t-}]$, and $p(I_{t})$ denotes the reinsurance premium corresponding to the reinsurance contract $I_{t}$. Remark that the contract $I_{t}$ may depend on the surplus level up to $t$, so it is stochastic.

In this paper, we consider the expected value premium principles for both insurance and reinsurance contracts which are given by
\[
p=(1+\de)\E [Z_1],\quad
p(I_t)=(1+\rho)\E_{R_{t-}} [I_t(Z_1)],\]
where $\de$, $\rho>0$ are called safety loadings of the insurance premium and reinsurance premium. This is fundamental to the insurance pricing as it stipulates that the insurer/reinsurer has a positive safety loading on the underwritten risk. In this case, we can rewrite \eqref{Rt3} as
\begin{equation*}
\d R_t=\big(-\ga+\rho\E_{R_{t-}}[Z_1-I_t(Z_1)]\big)\d t+\sqrt{\E_{R_{t-}}[(Z_1-I_t(Z_1))^2]}\;\d W_t,
\end{equation*}
with $\ga:=(\rho-\de)\E [Z_1]$. We impose that $\rho>\de$, i.e., $\ga>0$, to ensure that the reinsurance is non-cheap. If reinsurance is too cheap, that is $\rho<\de$, then the insurance company can simply eliminate its risk exposure by ceding all the incoming claims to reinsurance companies, reaping the profit of $\de-\rho>0$ with a zero ruin probability. Thanks to time rescaling, we will assume $\rho=1$ throughout this paper.

In this paper, we assume that the insurance company will pay out part of its surplus as dividends to its shareholders. Let $L_t$ be the cumulative dividend extracted from the surplus process until $t$, which is a non-decreasing c\`{a}dl\`{a}g (i.e., right continuous with left limits) process. It is chosen by the insurance company according to its surplus level. Then the new surplus process $\{R_s\}_{s\geq t-}$ beginning at time $t-$ with an initial value $x$ satisfies the following dynamics 
\begin{equation}\label{Rt}
\begin{cases}
\d R_s=\big(\!\!-\ga+\E_{R_{s-}} [H_s(Z_1)]\big)\d s+\sqrt{\E_{R_{s-}} [H_s^2(Z_1)]}\;\d W_s-\d L_s,\quad s\geq t,\bigskip\\
R_{t-}=x> 0,
\end{cases}
\end{equation}
thanks to $H(z)\equiv z-I(z)$.
This is also a c\`{a}dl\`{a}g process, which jumps at the same time as $L$ does with the same jump size in the opposite direction, namely $R_{s}-R_{s-}=-(L_{s}-L_{s-})$ for any $s\geq t$. Define the ruin time of the insurance company as
\begin{equation}\label{tau}
\deathtime:=\inf\big\{s\geq t\;|\;R_s\leq0\big\}.
\end{equation}
The insurance company is not allowed to pay out dividends more than the existing surplus, so $L_{s}-L_{s-}\leq R_{s-}$ at any time $s$. As a consequence, $R_{s}=R_{s-}-(L_{s}-L_{s-})\geq 0 $. In particular, it implies that the surplus of the insurance company is zero at the ruin time, namely $R_{\deathtime}=0$.

The objective of our optimal reinsurance and dividend-payout model is to find a retained loss policy $\cH^t=\{H_s\}_{s\geq t-}$ (or equivalently, a reinsurance policy $\cI^t=\{I_s\}_{s\geq t-}$) and a dividend-payout policy $\cL^t=\{L_s\}_{s\geq t-}$ to maximize the expectation of discounted cumulative dividend payouts until bankruptcy or a given maturity $T>0$ which arrives earlier for the insurance company. The value function of our problem is defined as
\begin{equation}\label{valuefunction}
V(x,t)=\sup\limits_{\cH^t,\cL^t}\E\Bigg[\int_{t-}^{T\wedge\deathtime} e^{-c(s-t)}\d L_s\;\bigg|\;R_{t-}=x\Bigg],\quad x>0, \quad 0\leq t\leq T,
\end{equation}
where $c$ is a positive discount factor and the retained loss function (or equivalently, the ceded function) is subject to the constraint
\[
0\leq H_s(Z)\leq Z,\quad s\in[t-,T].
\]
This is a mixed singular-classical control problem.

In the rest of this paper, we will investigate the value function and provide the optimal reinsurance and dividend-payout strategies for problem \eqref{valuefunction}.

%%%%%%%%%%%%%%%%%%%%%%%%%%%%%%%%%%%%%%%%%%%%%%%%%%%%%%%%%%%%%%%%%%%%%%%%%%%

\section{The HJB equation and Verification Theorem}\label{sec:HJB}
\setcounter{equation}{0}

We now study problem \eqref{valuefunction} by the dynamic programming principle. To this end, we introduce the following variational inequality
\begin{align}\label{v_pb00}
\begin{cases}
\displaystyle\min\big\{v_t-{\cal L}v,\;v_x-1 \big\}=0,\quad {\rm in}\quad \Q_T:=(0,+\infty)\times(0,T],\medskip\\
v(0,t)=0,\quad 0<t\leq T,\medskip\\
v(x,0)=x,\quad x>0,
\end{cases}
\end{align}
where
\begin{align}\label{calL}
{\cal L}v&:=\sup\limits_{H\in \aH}\left(\frac{ v_{xx}}{2}\int_0^\infty H(z)^2 \d F(z)+v_x\int_0^\infty H(z) \d F(z) \right)-\ga v_x-cv,\\
\aH &:=\big\{H:[0,\infty)\to[0,\infty)\;|\;0\leq H(z)\leq z\big\},\nonumber
\end{align}
and $F(\cdot)$ denotes the common cumulative distribution function of the claims with $F(0-)=0$ (due to the non-negativity of the claims).
The variational inequality \eqref{v_pb00} is indeed the (time-reversed) HJB equation for our problem \eqref{valuefunction}.

The classical approach to linking an HJB equation to a control problem like \eqref{valuefunction} is using the theory of viscosity solution (see \cite{YZ99}). But this theory requires one to show the uniqueness of the viscosity solution, which is an extremely challenging task for us. Moreover, this approach usually cannot provide a classical solution to the HJB equation. In this paper, we adopt a different approach. We first show that \eqref{v_pb00} has a classical solution by pure PDE methods, and then show that the solution is indeed the value function of problem \eqref{valuefunction} by a verification argument. The first result is shown by the penalty approximation method in PDE, so one can compute the value function as well as the optimal strategies for problem \eqref{valuefunction} by solving the approximation PDE numerically. This is an advantage of our method compared to other probabilistic arguments or the viscosity solution approach. 

Throughout this paper, we put the following technical assumption:
\begin{align}\label{distributionF}
\textrm{ $z^3(1-F(z))$ is bounded on $[0,\infty)$.}
\end{align}
Although this assumption may be relaxed slightly, we will not pursuit this direction in this paper. 
Please note that this assumption can be satisfied even if $F(\cdot)$ is discontinuous. 

\begin{theorem}\label{thm:v}
The variational inequality \eqref{v_pb00} has a unique solution $v\in C^{2,1}\big(\overline{\Q_T}\setminus \{(0,0)\}\big)\bigcap C\big(\overline{\Q_T}\big)$ that satisfies
\begin{align}\label{vx}
&v_x\geq1,\bigskip\\\label{vt}
&v_t\geq0,\bigskip\\\label{vxx}
&v_{xx}\leq 0,\bigskip\\\label{vxxx}
&v_{xxx}\geq 0\hbox{ in the weak sense},\bigskip\\\label{vxt}
&v_{xt}\geq 0,\bigskip\\\label{vga}
&\la v_x+v_{xx}\geq 0,
\end{align}
where $\la$ is the unique positive root of the function $f$ defined by \eqref{f}.
\end{theorem}
The proof is based on pure PDE methods, very delicate and long, so we leave it in Appendix \ref{sec:proof1}.

\begin{theorem}[Verification Theorem]
\label{thm:averi}
Suppose $v\in C^{2,1}\big(\overline{\Q_T}\setminus\{(0,0)\}\big)\bigcap C\big(\overline{\Q_T}\big)$ is increasing and concave w.r.t. $x$ and satisfies \eqref{v_pb00}. Then the value function of the optimal reinsurance and dividend-payout problem \eqref{valuefunction} is given by
\be\label{V1}
V(x,t)=v(x,T-t).
\ee
Moreover, the optimal ceded loss policy $I^{*}$ is given by a feedback control of the loss and surplus:
\[I_{s}^{*}(z, R^{*}_{s-})=\max\left\{0,\;z+\frac{v_{x}(R^{*}_{s-},T-s)}{v_{xx}(R^{*}_{s-},T-s)}\right\},\]
and the optimal dividend-payout strategy $L^*$ for problem \eqref{valuefunction} is such that 
\[\begin{cases}
L^*_s-L^*_{s-}=R^{*}_{s-}-d^{*}(T-s), & \hbox{if}\quad R^{*}_{s-}>d^{*}(T-s);\\
L^*_s-L^*_{s-}=0, & \hbox{if}\quad R^{*}_{s-}\leq d^{*}(T-s),
\end{cases}\]
where the payout free boundary $d^*$ is given by
\[d^{*}(s)=\inf\{x\geq0\;|\;v_x(x,s)=1\},\quad s\in[0,T],\]
with the convention $\inf\emptyset = \infty$.
\end{theorem}
The proof is fairly standard and given in Appendix \ref{sec:verify}.

From now on, we fix $v$ as in \thmref{thm:v}. By Verification \thmref{thm:averi}, it completely characterizes the value function of the optimal reinsurance and dividend-payout problem \eqref{valuefunction}. In the next section, we study the properties of the optimal reinsurance and dividend-payout strategies.

\section{Optimal Strategies}
In the previous section, we have resolved the existence and uniqueness issues for the time-reversed HJB equation \eqref{v_pb00}. In the following part, we investigate the optimal strategies for problem \eqref{valuefunction}. We will use the notation in Appendix \ref{sec:proof1} including the constants $\mu_1$ and $\mu_2$, functions $A(\cdot)$, $B(\cdot)$, and the operator ${\cal T}$ that are defined by \eqref{defmu}, \eqref{defA}, \eqref{defB}, and \eqref{calT}, respectively.

If $\ga\geq\mu_1$, then it is easy to check that $v\equiv x$ is the solution to problem \eqref{v_pb00}. In this case one can see that the drift of the surplus process in \eqref{Rt} is either negative if $\E[I_{t}(Z_{1})]>0$ or 0 otherwise, which means the insurance company has no incentive to survive. Consequently, the optimal policy is to pay out all reserves as dividends and let the company be bankrupt immediately. This gives the optimal value $x$ for problem \eqref{valuefunction}. The problem is trivial in this case. Hence, in the rest part of this section, we assume $0<\ga<\mu_1$.

We first study the optimal dividend-payout strategy and then the optimal reinsurance strategy in the subsequent two sections.

\subsection{Optimal dividend-payout strategy}\label{sec:D}
\setcounter{equation}{0}

To investigate the optimal dividend-payout strategy, we divide the whole domain $\TQ:=(0,+\infty)\times[0,T)$ into a {\bf dividend-payout region}
\[
\cD= \left\{(x,t)\in \TQ \;\Big|\; v_{x}(x,\vt)=1\right\}
\]
and a {\bf no-dividend-payout region}
\[
\ND= \left\{(x,t)\in \TQ \;\Big|\; v_{x}(x,\vt)>1\right\}.
\]
Here and hereafter we use the notation $\vt:=T-t$.

Since $v_{x}\geq 1$ and $v_{xx}\leq0$, we can express them as
\[
\cD=\left\{(x,t)\in \TQ \;\Big|\;x\geq d(\vt)\right\},\quad \ND=\left\{(x,t)\in \TQ \;\Big|\;x<d(\vt)\right\}, 
\]
where $d(\cdot)$ is the {\bf dividend-payout boundary}, defined by
\[
d(\vt)=\inf\{x\geq0\;|\;v_x(x,\vt)=1\},\quad \vt>0.
\]

In the following part, we come to show the boundary $d(\cdot)$ is uniformly upper bounded by an explicit given constant. To this end, we will construct a function $\wu(x)$ such that $\wu(x)\geq v_x(x,t)$, then clearly $\inf\{x\geq0\;|\;\wu(x)=1\}$ provides a uniformly upper bound for $d(\cdot)$.

First, we show that $v_x(0,t)$ is uniformly upper bounded. For this, we construct a function
\[
\wv(x):=
\begin{cases}
C_1(1-e^{-\frac{x}{\ga}}), & 0<x\leq x_1, \\
C_2+x-x_1, & x> x_1
\end{cases}\]
where
\begin{align*}
C_{1}=\frac{\mu_{1}}{c}+\ga>0,\quad
C_{2}=\frac{\mu_{1}}{c}>0,\quad
x_1=\ga\ln\frac{C_1}{\ga}>0.
\end{align*}
It is easy to check
$\wv(x_1-)=\wv(x_1+)$ and $\wv_x(x_1-)=\wv_x(x_1+)$, so $\wv\in C^{1}(0,\infty)$. It is easily seen that $\wv_x$ is continuous and decreasing, so $\wv$ is a concave function. When $0<x\leq x_1$,
\begin{align*}
\wv_t-{\cal L}\wv
&=-\int_0^\infty\sup\limits_{0\leq h\leq z}\Big(\frac{1}{2}h^2 \wv_{xx}+h \wv_x\Big)\d F(z)+\ga \wv_x+c\wv\bigskip\\
&\geq-\sup\limits_{0\leq h< \infty}\Big(\frac{1}{2}h^2\wv_{xx}+h \wv_x \Big)+\ga \wv_x=\frac{\wv_x^2}{2\wv_{xx}}+\ga \wv_x=\frac{C_1}{2} e^{-\frac{x}{\ga}}>0,
\end{align*}
and when $x> x_1$,
\[
\wv_t-{\cal L}\wv=-\mu_{1}+\ga+c(C_2+x-x_1)=\ga+c(x-x_1)> 0.
\]
Therefore, $\wv\in W^{2,1}_{p}(\Q_T)$ is a super solution to problem \eqref{v_pb00}. Since $\wv(0)=v(0,t)=0$, we obtain that $v_x(0,t)\leq \wv_x(0)=C_1/\ga$ by the comparison principle.

\begin{remark}
We would like to provide a probabilistic representation for $\wv$. Consider
\begin{equation} \label{newcost}
\sup\limits_{\{h_s\}_{s\geq t-},\{L_s\}_{s\geq t-}}\E\Bigg[\int_{t-}^{T\wedge\deathtime} \d L_s\;\bigg|\;R_{t-}=x\Bigg],
\end{equation}
for the following new surplus
\begin{equation}\label{newsur}
\begin{cases}
\d R_s=\big( -{\frac{\ga}{2}}+h_s \big)\d s+h_s\;\d W_s-\d L_s,\quad s\geq t,\bigskip\\
R_{t-}=x> 0.
\end{cases}
\end{equation}
By establishing an analog of Verification Theorem \ref{thm:averi}, one can show that $\wv$ is the value function for this problem.

For any reinsurance and dividend-payout strategies $\{H_s\}_{s\geq t-}$ and $\{L_s\}_{s\geq t-}$ for the surplus \eqref{Rt}, we take the same dividend-payout strategy $\{L_s\}_{s\geq t-}$ and choose $\{h_s\}_{s\geq t-}=\{\sqrt{\E_{R_{s-}} [H_s^2(Z_1)]}\}_{s\geq t-}$ for the new surplus \eqref{newsur}, then one can easily show the new surplus is always higher than the surplus \eqref{Rt} by virtue of
\begin{align*}
-{\frac{\ga}{2}}+h_s =-{\frac{\ga}{2}}+\sqrt{\E_{R_{s-}} [H_s^2(Z_1)]}
> -\ga+\E_{R_{s-}} [H_s(Z_1)].
\end{align*}
Meanwhile, the new cost functional \eqref{newcost} is clearly no less than the one in \eqref{valuefunction}, so we conclude $\wv\geq V$.
\end{remark}

Now, we are ready to construct an upper bound function for $v_x(x,t)$. To this end, let
\[
\wu(x):=
\begin{cases}
C_3(x_2-x)^2+1, & 0<x\leq x_2, \\
1, & x> x_2,
\end{cases}\]
where
\[C_3=\frac{c^{2}}{c\mu_2+\ga^2}>0,\quad { x_2:=\sqrt{\frac{C_2}{C_{3}\ga}}=\sqrt{\frac{1}{\ga c^{3}}\left(\mu_1+c\ga\right)\left(c\mu_2+\ga^2 \right)}}>0.\]
Clearly, $\wu$ is convex and $\wu \in W^{2,1}_{p}(\Q_T)$. If $0<x<x_2$, then $\wu_x\leq 0$. By \eqref{A_b} and using the elementary inequality $x^{2}-2xy\geq-y^{2}$, we obtain
\begin{align*}
\wu_t-{\cal T}\wu
&\geq-\frac{1}{2}\mu_2 \wu_{xx}+\ga \wu_x+c\wu\bigskip\\
&=C_3\left(-\mu_2-2\ga (x_2-x)+c(x_2-x)^2\right)+c\bigskip\\
&\geq C_3\left(-\mu_2-\ga^2/c \right)+c=0.
\end{align*}
If $x>x_{2}$, then $\wu_t-{\cal T}\wu=c>0$. As $\wu(0)=C_1/\ga\geq v_x(0,t)=u(0,t)$, by the comparison principle, we conclude $\wu\geq u$. Therefore, $x_2=\inf\{x\geq0\;|\;\wu(x)=1\}$ is a constant upper bound for $d(\cdot)$.

Summarizing the above results, the dividend-payout boundary is completely characterized in the following theorem.
\begin{theorem}\label{thm:h}
The dividend-payout boundary $d(\vt)$ is continuous and increasing in $\vt$, and satisfies
\be\label{x2}
d(0+)=0<d(\vt)\leq d(\infty)\leq x_{2}=\sqrt{\frac{1}{\ga c^{3}}\left(\mu_1+c\ga\right)\left(c\mu_2+\ga^2 \right)},
\ee
where $d(0+)=\lim\limits_{\vt\rt 0+}d(\vt)$ and $d(\infty):=\lim\limits_{\vt\rt+\infty}d(\vt)$.
Furthermore, if $Z_{1}$ is a bounded random variable, then $d(\vt)\in C^\infty (0,T)$.
\end{theorem}
See Figure \ref{fig1} for an illustration of the dividend-payout barrier $d(\vt)$ as well as dividend-payout region $\cD$ and no-dividend-payout region $\ND$.
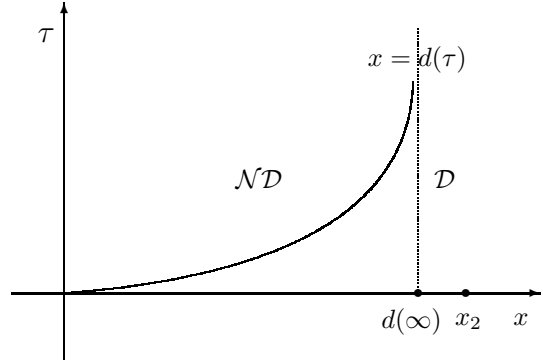
\begin{figure}[H]
\begin{center}
\begin{picture}(200,125)(0,0)
\put(0,10){\vector(1,0){200}} \put(20,-15){\vector(0,1){135}}
\qbezier[80](154, 10)(154, 60)(154, 110)
\put(10,105){$\vt$}
\put(190,-2){\footnotesize $x$}
\put(140,-3){\footnotesize $d(\infty)$}
\put(168,-2){\footnotesize $x_2$}
\put(85,50){\footnotesize $\ND$}
\put(160,50){\footnotesize $\cD$}
\put(154, 10){\circle*{3}}
\put(172, 10){\circle*{3}}
\qbezier(20,10)(150,20)(152,90)\put(135,96){\footnotesize $x=d(\vt)$}
\end{picture} \bigskip
\caption{The dividend-payout barrier $x=d(\vt)$ divides the surplus-time space into a dividend-payout region $\cD$ and a no-dividend-payout region $\ND$.} \label{fig1}\end{center}
\end{figure}
\begin{proof}
The monotone property of $d(\cdot)$ is due to \eqref{vxt}. We now prove the continuity. Suppose there exists a $\vt_0$ such that $d(\vt_0-)<d(\vt_0+)$. Then for any $x\in (d(\vt_0-),d(\vt_0+))$, by the continuity of $v_x$ and the monotonicity of $v_x$ and $d(\cdot)$,
\begin{align*}
v_x( x,\vt_0)=\lim_{\ep\to 0+}v_x(x,\vt_0-\ep)
\leq \lim_{\ep\to 0+}v_x(d(\vt_0-\ep),\vt_0-\ep)=1.
\end{align*}
Thanks to $v_x\geq 1$, we have $v_x(x,\vt_0)=1$. Hence, for any $x\in (d(\vt_0-), d(\vt_0+))$, $v_{xx}(x,\vt_0)=0$. Consequently, we infer from \eqref{v_pb00} that $v_t(x, \vt_0)=-\mu_{1}+\ga-cv(x, \vt_0)$ and thus $v_{tx}(x,\vt_0)=-cv_x(x, \vt_0)=-c<0$, which contradicts \eqref{vxt}. So we conclude $d(\cdot)$ is continuous. Similarly, we can prove $\lim\limits_{\vt\rt 0+}d(\vt)=0$.

Now we prove $d(\cdot)>0$. Suppose not, then, by monotonicity, there exists a $\vt_0>0$ such that $d(\vt)=0$ for all $0<\vt<\vt_0$. This implies $v_x\equiv 1$ or $v\equiv x$ for $0<\vt<\vt_0$. Denote $\wv(x,\vt)=v(x,\vt-\vt_0)$, then both $\wv$ and $v$ satisfy \eqref{v_pb00}. By the uniqueness, we get $\wv\equiv v$, namely $v\equiv x$ for $0<\vt<2\vt_0$. By the mathematical induction, we have $v\equiv x$ for all $\vt>0$, which leads to
\[v_t-{\cal L}v=-\mu_1+\ga+cx.\]
This contradicts \eqref{v_pb00} for any sufficiently small $x$ as $\ga<\mu_1$. 

It is only left to show the smoothness of $d(\cdot)$ when $Z_{1}$ is bounded. Suppose $F(\hatz)=1$ for some $\hatz>0$. Then by \eqref{defA} and \eqref{defB},
\be\label{cd1}
A(y)=\frac{1}{2}\mu_2,\quad B(y)=\mu_1,\quad y^{-1}\geq \hatz.
\ee
For any $\vt_0>0$, as $v_x$ and $v_{xx}$ are continuous, we have $v_x=1$ and $v_{xx}=0$ at $(d(\vt_0),\vt_0)$. This implies $v_x+\hatz v_{xx}>0$ in some neighborhood ${\cal B}$ of $(d(\vt_0),\vt_0)$. By \eqref{u_pb} and \eqref{cd1}, we see $u=v_x$ satisfies
\[
\min\Big\{u_t-\frac{\mu_{2}}{2}u_{xx}-(\mu_{1}-\ga) u_x+cu,\;u-1\Big\}=0,\quad (x,\vt)\in {\cal B}.
\]
Because the coefficients are constants in the above equation, using the method in \cite{Fr75}, we can prove $d(\vt)\in C^\infty$ at some neighborhood of $\vt_0$. Since $\vt_0$ is arbitrarily chosen, we conclude that $d(t)\in C^\infty(0,T)$.
\end{proof}

In view of the original optimal reinsurance and dividend-payout problem \eqref{valuefunction}, by \thmref{thm:averi}, the optimal dividend-payout policy $L^*_s$ is the local time of the corresponding reserve process $R_s^*$ at the level $d(\vs)$, namely
\be\label{L*}
\begin{cases}
L^*_s-L^*_{s-}=R^*_{s-}-{d(\vs)}, & \hbox{if}\quad R^*_{s-}>d(\vs);\\
\d L^*_s=0, & \hbox{if}\quad R^*_{s-}\leq d(\vs).
\end{cases}
\ee 
Under this policy, the reserve process $R_s^*$ is always continuous and no more than $d(\vs)$, except for the initial time $t$. When the insurance surplus $R^*_{t-}$ is above the threshold $d(T-t)$, the insurance company should pay out the reserves of an amount $R^*_{t-}-d(T-t)$ as dividends to its shareholders at the initial time $t$; otherwise, pay out nothing. Indeed the accumulated dividends increase with the local time at the boundary.

\subsection{Optimal reinsurance strategy}\label{sec:R}
In this section we study the optimal reinsurance strategy. Recall that we have confirmed $V(x,t)=v(x,\vt)$ by \thmref{thm:averi}.

For the insurance company, if its current state is $(x,t)$, by \eqref{vx}, \eqref{vxx} and \eqref{h*}, the corresponding optimal retained function is
\be\label{H^}
z\mapsto \WH(z,x,\vt):= 
\begin{cases}
-\frac{v_x}{v_{xx}}(x,\vt), & {\rm if}\;-\frac{v_{xx}}{v_{x}}(x,\vt)> \frac{1}{z};\bigskip\\
z, & {\rm otherwise},
\end{cases}
\ee
and the optimal reinsurance function is
\[z\mapsto \WI(z,x,\vt):=z-\WH(z,x,\vt).\]
The reinsurance contract depends on both the current insurance surplus $x$ and time $t$. In this section, we will discuss the behavior of them.

For each $z>0$, depending on whether $\WI$ is zero, we divide the surplus-time space into a {\bf reinsurance-covered region}
\[
\PR_{z}=\left\{(x,t)\in \TQ \;\Big|\; \WI(z,x,\vt)>0\right\}
\]
and a {\bf reinsurance-uncovered region}
\[
\NR_{z}=\left\{(x,t)\in \TQ \;\Big|\;\WI(z,x,\vt)=0\right\}.
\]
By \eqref{H^}, they can also expressed as
\[
\PR_{z}=\left\{(x,t)\in \TQ \;\Big|\;-\frac{v_{xx}}{v_{x}}(x,\vt)> \frac{1}{z}\right\}, \quad
\NR_{z}=\left\{(x,t)\in \TQ \;\Big|\;-\frac{v_{xx}}{v_{x}}(x,\vt)\leq \frac{1}{z}\right\}.
\]
Our above discussion shows that $\cD\subseteq \NR_{z}$ and $\PR_{z}\subseteq \ND$.

\begin{lemma}\label{prop:vR}
We have
\begin{align}\label{vR}
v_x,\;v_t\in C^{2,1}\big(\ND\big).
\end{align}
Furthermore, $v_{xx}<0$ if $(x,t)\in\ND$. Consequently, $\WI(z,x,\vt)=\max\left\{z+v_{x}/v_{xx}(x,\vt),\; 0\right\}$ for $z>0$ and $(x,t)\in\ND$.
\end{lemma}
\begin{proof}
The proof is given in Appendix \ref{sec:proof2}.
\end{proof}

This results implies that $\WI$ is an increasing function w.r.t $z$, which makes a perfect financial meaning that the insurance company would get more compensations when a larger claim arises. Another consequence is that 
\[\WI(z,x,\vt)<z, \quad (x,t)\in\ND.\]
The optimal reserve process $R^*$ is always continuous and no more than the dividend-payout barrier, except for the initial time, so we always have $\WI(z,R_{t-}^*,\vt)<z$. In any case, the insurance company should not cede all the risks; it has to bear some risks by itself. This is due to our assumption that the reinsurance is non-cheap, $\ga>0$.

\begin{lemma}\label{prop:ax}
We have
\begin{align}\label{A0}
\left(-\frac{v_{x}}{v_{xx}}\right)\Bigg|_{x=0}&=\frac{1}{\la},\\\label{A}
-\frac{v_{x}}{v_{xx}}(x,\vt)&>\frac{1}{\la} \quad{\rm if }\quad (x,t)\in\ND,\\\label{hx}
\p_x\left(-\frac{v_{x}}{v_{xx}}\right)(x,\vt)&\geq 2c \quad{\rm if}\quad (x,t)\in\ND,
\end{align}
where $\la$ is the unique positive root of the function $f$ defined by \eqref{f}.
\end{lemma}
\begin{proof}
The proof is given in Appendix \ref{sec:proof3}.
\end{proof}

From \eqref{A} we see that, if $(x,t)\in\ND$, then
\begin{align}\label{wiprop}\WI(z,x,\vt)=\max\left\{z+v_x/v_{xx}(x,\vt),\; 0\right\}\leq \max\left\{z-1/\la,\; 0\right\}.
\end{align}
So $\WI(z,x,\vt)=0$ when $z\leq \frac{1}{\la}$; the insurance company should bear all the incoming claims below $\frac{1}{\la}$ by itself. In other words, any arising claim below $\frac{1}{\la}$ need not be shared with reinsurance companies.

By \eqref{A} and \eqref{hx}, we have $-\frac{v_{xx}}{v_{x}}$ is a positive strictly decreasing function in $x$ in $\ND$. Because $\PR_{z}\subseteq \ND$, 
\[\PR_{z}=\Big\{(x,t)\in \TQ \;\Big|\;x<K(z,\vt) \Big\}, \quad
\NR_{z}=\Big\{(x,t)\in \TQ \;\Big|\;x\geq K(z,\vt) \Big\}, \]
where $K(z,\vt)$ is the reinsurance boundary, given by
\[K(z,\vt):=\sup\left\{x>0\;\bigg|\;-\frac{v_{xx}}{v_{x}}(x,\vt)> \frac{1}{z}\right\}.\]
Because $\PR_{z}\subseteq \ND$, we see $K(z,\vt)\leq d(\vt)$. Moreover, when $z\leq 1/\la$, by \eqref{wiprop} we have $ \PR_{z}=\emptyset$, so $K(z,\vt)=0$.

For each $z>0$, define the overlapping of the reinsurance-uncovered region and non-dividend region as {\bf non-action region}
\[\NA_{z}=\NR_{z}\bigcap\ND=\Big\{(x,t)\in \TQ \;\Big|\; K(\hatz,\vt)\leq x<d(\vt) \Big\}.\]
In this region, the insurance company should not cede risk of magnitude $z$ or pay out dividends. We now show this region is always non-empty. This is equivalent to showing $K(z,\vt)\neq d(\vt)$ as $K(z,\vt)\leq d(\vt)$.
Note that $d(\vt)=\inf\{x\geq0\;|\;v_x(x,\vt)=1\}$ and $v\in C^{2,1}\big(\overline{\Q_T}\setminus \{(0,0)\}\big)\bigcap C\big(\overline{\Q_T}\big)$, so $(v_{x}+zv_{xx})\big|_{(d(\vt),\vt)}=1$. But evidently $(v_{x}+zv_{xx})\big|_{(K(z,\vt),\vt)}=0$, so $K(z,\vt)\neq d(\vt)$.

Combining the above results, we conclude that, for each magnitude of risk $z>0$, the reinsurance and dividend-payout barriers divide the surplus-time space into three non-overlapping regions: a (possible empty) reinsurance-covered region, a non-action region and a dividend-payout region, in an increasing order of the surplus.
This is illustrated in Figure \ref{fig2}.
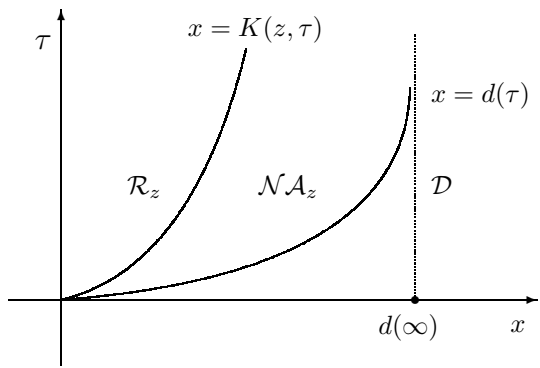
\begin{figure}[H]
\begin{center}
\begin{picture}(200,125)(0,0)
\put(0,10){\vector(1,0){200}} \put(20,-15){\vector(0,1){135}}
\qbezier[80](154, 10)(154, 60)(154, 110)
\put(10,105){$\vt$}
\put(190,-2){\footnotesize $x$}
\put(140,-3){\footnotesize $d(\infty)$}
\qbezier(20,10)(70,22)(90,105)
\put(68,110){\footnotesize $x=K(z,\vt)$}
\put(45,50){\footnotesize $\PR_{z}$}

\put(95,50){\footnotesize $\NA_{z}$}
\put(160,50){\footnotesize $\cD$}
\put(154, 10){\circle*{3}}
\qbezier(20,10)(150,20)(152,90)
\put(160,85){\footnotesize $x=d(\vt)$}
\end{picture} \bigskip
\caption{The surplus-time space is divided into three non-overlapping regions: $\PR_{z}$, $\NA_{z}$ and $\cD$ by the reinsurance barrier $x=K(z,\vt)$ and dividend-payout barrier $x=d(\vt)$. Moreover, $\ND=\PR_{z}\bigcup\NA_{z}$ and the $\vt$-axis can be regarded as $x=K(1/\la,\vt)$ since $\PR_{1/\la}=\emptyset$.
} \label{fig2}\end{center}
\end{figure}

Economically speaking, if the surplus level of the insurance company is relatively low ($x<K(z,\vt)$), then the company must cede risk of magnitude $z$ to other reinsurance companies to avoid bankruptcy. If the surplus level is medium ($K(z,\vt)\leq x\leq d(\vt)$), then the insurance company can cover the claim of magnitude $z$ by itself, but has not enough reserves to pay out as dividends, so no actions will be taken. If the surplus level is very high ($x>d(\vt)$), the company should pay out its extra reserves as dividends to its shareholders. It is never optimal for the insurance company to buy reinsurance contracts and pay out reserves as dividends, simultaneously.

By \lemref{prop:ax}, $K(z,\vt)$ is increasing and the reinsurance-covered region $\PR_{z}$ is getting larger as the risk magnitude $z$ increases. Economically speaking, if a risk magnitude is covered by reinsurance contracts, then any higher magnitude of risk should be covered as well.

Define a {\bf uniform non-action region}
\[\NA_{0}=\bigcap_{z>0}\NA_{z}.\]
In this region the insurance company should not cede any risk or pay out dividends.

If $\hatz=\esssup Z_{1}<\infty$, then $K(z,\vt)\leq K(\hatz,\vt)$. Therefore,
\[\bigcup_{z>0}\PR_{z}=\Big\{(x,t)\in \TQ \;\Big|\;x< K(\hatz,\vt) \Big\}=\PR_{\hatz}.\]
As $ K(\hatz,\vt)<d(\vt)$, we see that
\[\NA_{0}=\Big\{(x,t)\in \TQ \;\Big|\; K(\hatz,\vt)\leq x<d(\vt) \Big\}=\NA_{\hatz}\neq \emptyset.\]
This is illustrated in Figure \ref{fig3}.
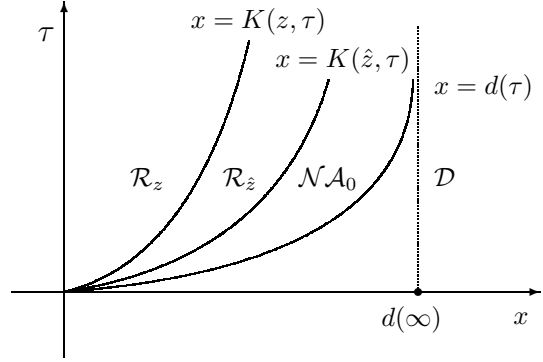
\begin{figure}[H]
\begin{center}
\begin{picture}(200,125)(0,0)
\put(0,10){\vector(1,0){200}} \put(20,-15){\vector(0,1){135}}
\qbezier[80](154, 10)(154, 60)(154, 110)
\put(10,105){$\vt$}
\put(190,-2){\footnotesize $x$}
\put(140,-3){\footnotesize $d(\infty)$}
\qbezier(20,10)(70,22)(90,105)
\put(68,110){\footnotesize $x=K(z,\vt)$}
\put(45,50){\footnotesize $\PR_{z}$}
\qbezier(20,10)(100,22)(120,90)
\put(100,95){\footnotesize $x=K(\hatz,\vt)$}
\put(80,50){\footnotesize $\PR_{\hatz}$}
\put(109,50){\footnotesize $\NA_{0}$}

\put(160,50){\footnotesize $\cD$}
\put(154, 10){\circle*{3}}
\qbezier(20,10)(150,20)(152,90)
\put(160,85){\footnotesize $x=d(\vt)$}
\end{picture} \bigskip
\caption{As $z$ is increasing, the region $\PR_{z}$ is extending to $\PR_{\hatz}$ and the curve $x=K(z,\vt)$ is moving right up to $x=K(\hatz,\vt)$. The region $\NA_{0}=\NA_{\hatz}$ is non-empty when $Z_{1}$ is bounded.} \label{fig3}\end{center}
\end{figure}
Economically speaking, when the potential claims are bounded in magnitude $(\leq \hatz)$ and the surplus level is relatively high ($x\geq K(\hatz,\vt)$), then the insurance company can bear the claims by itself. 

Next, by \eqref{hx}, we have
\[\p_{x}\WI(z,x,\vt)\leq -2c \quad{\rm in}\quad \PR_{z}.\]
This means the insurance company should significantly reduce its purchase of reinsurance contracts as the surplus increases. If its surplus level is very high $x\geq \frac{z}{2c}-\frac{1}{2c\la}$, then we claim that there is no need to buy reinsurance which covers the risk of magnitude $z$, that is, $(x,t)\notin\PR_{z}$. In fact, if $z\leq 1/\la$, then there is nothing to prove since $\PR_{z}=\emptyset$. Otherwise, suppose $(x,t)\in\PR_{z}$. As $(x,t) \in\ND$, by \eqref{hx}, \eqref{wiprop} and the mean value theorem, we have
\[ \WI(z,x,\vt)\leq \WI(z,0+,\vt)-2cx\leq z-1/\la-2cx\leq 0,\]
which contradicts $(x,t)\in\PR_{z}$. Therefore, we conclude that
\[\PR_{z}\subseteq \left\{(x,t)\in \TQ \;\bigg|\;x< \frac{z}{2c}-\frac{1}{2c\la}\right\},\]
and consequently,
\[K(z,\vt)\leq \frac{z}{2c}-\frac{1}{2c\la}.\]
This is illustrated in Figure \ref{fig4}.
\begin{figure}[H]
\begin{center}
\begin{picture}(220,160)(0,0)
\put(0,10){\vector(1,0){210}} \put(20,-10){\vector(0,1){160}}
\qbezier[80](154, 10)(154, 60)(154, 110)
\put(10,130){$\vt$}
\put(190,-2){\footnotesize $x$}
\put(140,-3){\footnotesize $d(\infty)$}
\qbezier(20,10)(70,22)(90,105)
\put(68,110){\footnotesize $x=K(z,\vt)$}
\put(45,50){\footnotesize $\PR_{z}$}
\qbezier[100](110, 10)(110, 68)(110, 126)
\put(85,135){\footnotesize $x=\frac{z}{2c}-\frac{1}{2c\la}$}

\put(160,50){\footnotesize $\cD$}
\put(154, 10){\circle*{3}}
\qbezier(20,10)(150,20)(152,90)
\put(160,85){\footnotesize $x=d(\vt)$}
\end{picture} \bigskip
\caption{The line $x=\frac{z}{2c}-\frac{1}{2c\la}$ gives a universal upper bound for the reinsurance barrier $K(z,\vt)$ and the reinsurance-covered region $\PR_{z}$ when $z>1/\la$.} \label{fig4}\end{center}
\end{figure}
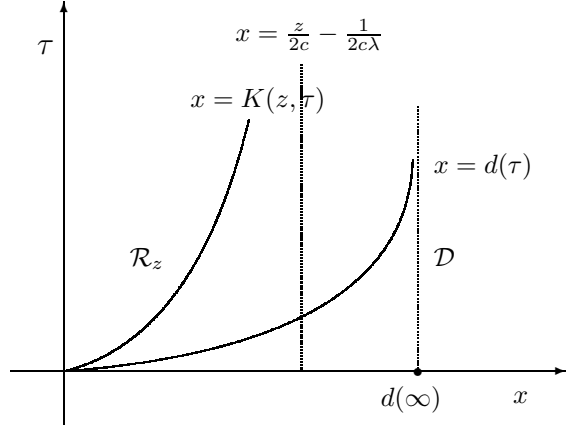
Economically speaking, the insurance company only shares sufficiently high magnitude of risk with other reinsurance companies.

In the rest part we assume the claims follows a discrete probability distribution given by
\[
\P(Z_1=z_j)=p_j>0,\quad j=1, 2,..., N,
\]
with
\[
0<z_1<z_2<\cdots<z_{i_0-1}\leq 1/\la<z_{i_0}<\cdots<z_N\quad{\rm and}\quad \sum_{j=1}^{N} p_j=1.
\]

Define the ${j^{th}}$ {\bf reinsurance boundary} as
\[
K_{j}(\vt):=\inf\left\{x>0\;\bigg|\;-\frac{v_{xx}}{v_x}(x,\vt)\leq \frac{1}{z_{j}}\right\},\quad j=1,2,...,N.
\]
And define the ${j^{th}}$ {\bf reinsurance-covered region} as
\[
\PR_{z}^{j}=\left\{(x,t)\in \TQ \;\Big|\; \WI(z_{j},x,\vt)>0\right\}.
\]
Thanks to \lemref{prop:ax}, the properties of these reinsurance boundaries are given in the following result.
\begin{theorem}\label{thm:r}
The reinsurance boundaries $K_j(\cdot)$, $j=1,2,\cdots, N$, are all continuously differentiable in time. Moreover,
\begin{align}\label{K1}
& K_j(\vt)=0\quad \hbox{for each}\quad j=1,2,...,i_0-1;\bigskip\\\label{K2}
&0<K_{i_0}(\vt)< \frac{z_{i_0}-\frac{1}{\la}}{2c};\bigskip\\\label{K3}
&0<K_j(\vt)-K_{j-1}(\vt)< \frac{z_j-z_{j-1}}{2c} \quad \hbox{for each}\quad j=i_0+1,...,N;\bigskip\\\label{K4}
&\lim\limits_{\vt\rt 0+}K_j(\vt)=0.
\end{align}
\end{theorem}

\begin{proof}
By \eqref{A0}, \eqref{hx} and $0<z_1<z_2<...<z_{i_0-1}\leq 1/\la$, we have \eqref{K1}. Note that $\left(-\frac{v_x}{v_{xx}}\right)(K_{i_0}(t),t)=z_{i_0}>\frac{1}{\la}$. Since $-\frac{v_x}{v_{xx}}$ is continuous in ${\cal ND}$, we obtain
$ K_{i_0}(t)>0.$ Thanks to the mean-value theorem and \eqref{hx},
\[
z_1-\frac{1}{\la}=\left(-\frac{v_x}{v_{xx}}\right)(K_1(t),t)-\left(-\frac{v_x}{v_{xx}}\right)(0,t)> 2c K_1(t).
\]
It gives \eqref{K2}. The proof of \eqref{K3} is similar. The fact that $K_1(t)<...<K_N(t)<d(t)$ and $\lim\limits_{t\rt 0+}d(t)=0$ gives \eqref{K4}.

Now, we prove $K_j(t)\in C^1((0,T])$ for $j\geq i_0$. By \lemref{prop:vR} and $v_{xx}<0$ in ${\cal ND}$, we see that $-\frac{v_x}{v_{xx}}\in C^{1,1}({\cal ND})$. Because $\p_x\left(-\frac{v_x}{v_{xx}}\right)> 2c$ in ${\cal ND}$, the implicit function theorem implies that $K_j(t)\in C^1((0,T])$.
\end{proof}

Figure \ref{fig5} illustrates this result.
\begin{figure}[H]
\begin{center}
\begin{picture}(220,140)(0,0)
\put(0,10){\vector(1,0){210}} \put(20,-10){\vector(0,1){135}}
\qbezier[80](154, 10)(154, 60)(154, 110)
\put(10,110){$\vt$}
\put(190,-2){\footnotesize $x$}
\put(140,-3){\footnotesize $d(\infty)$}
\qbezier(20,10)(40,12)(45,90)\put(34,95){\footnotesize $K_{i_0}(\vt)$}
\qbezier(20,10)(72,20)(84,90)\put(63,95){\footnotesize $K_{i_0+1}(\vt)$}
\qbezier(20,10)(110,20)(120,90)\put(108,95){\footnotesize $K_N(\vt)$}
\put(85, 60){$\cdots$}
%\put(90, 94){$\cdots$}
\put(23, 60){$\PR_{z}^{i_0}$}
\put(49, 60){$\PR_{z}^{i_0+1}$}
\put(118, 60){\footnotesize $\NA_{0}$}
\put(160, 60){\footnotesize $\cD$}
\put(154, 10){\circle*{3}}
\qbezier(20,10)(150,20)(152,90)
\put(140, 95){\footnotesize $d(\vt)$}
\end{picture} \bigskip
\caption{The reinsurance boundaries $K_j$ and reinsurance-covered regions $\PR_{z}^{j}$ are increasing in $j$. The uniform non-action region $\NA_{0}$ is $\NA_{z_{N}}$.
} \label{fig5}\end{center}
\end{figure}
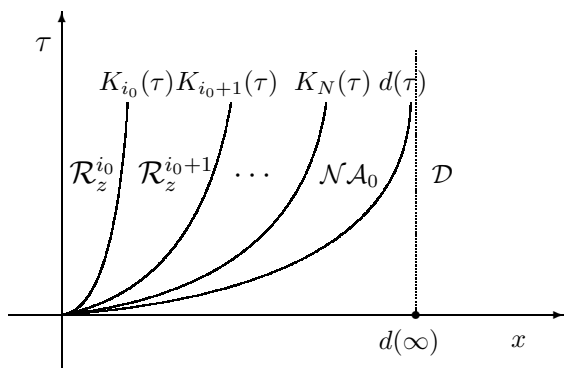

Economically speaking, the insurance company should only cede the claims $z_{j}, z_{j+1},\cdots z_{N}$, if the surplus level is in $\PR_{z}^{j}$. \bigskip
\par
\textbf{Acknowledgment.} The authors would like to thank the Area Editor Prof.\;Jose Blanchet, the anonymous Associate Editor and two referees for their valuable comments and suggestions, which greatly improve the previous two versions of the manuscript.

\newpage
\begin{appendices}

\section{Proof of \thmref{thm:v}.}\label{sec:proof1}
\setcounter{equation}{0}
Thanks to \eqref{distributionF}, $z(1-F(z))$ is dominated by $z^{-2}$ near infinity, hence, by Fubini's theorem and the dominated convergence theorem,
\[\int_0^{\infty} z^2\d F(z)=2\int_0^{\infty} z(1-F(z))\d z <\infty.\]
Therefore, the first and second moments of the claims are finite, that is, 
\begin{align}\label{defmu}
\mu_1:=\int_0^\infty z\d F(z)<\infty, \quad \mu_2:=\int_0^\infty z^2\d F(z)<\infty.
\end{align}

To solve \eqref{v_pb00}, we first solve the optimization problem in the operator ${\cal L}$, namely
\begin{align*}
\sup\limits_{H\in \aH}\Big(\frac{ v_{xx}}{2}\int_0^\infty H(z)^2 \d F(z)+v_x\int_0^\infty H(z) \d F(z) \Big),
\end{align*}
or
\begin{align*}
\sup\limits_{H\in \aH}\int_0^\infty\Big(\frac{1}{2} H(z)^2v_{xx}+H(z)v_{x} \Big)\d F(z),
\end{align*}
which is obviously equivalent to
\be\label{sup}
\int_0^\infty\sup\limits_{0\leq h\leq z}\Big(\frac{1}{2}h^2 v_{xx}+h v_x\Big)\d F(z).
\ee
Define
\begin{align}\label{h*}
h^*(z,y):=\argmax_{0\leq h\leq z}\left(-\frac{1}{2}h^2y+h \right)=
\begin{cases}
\min\{z,y^{-1}\}, & {\rm if}\; 0<y<\infty; \bigskip\\
z, & {\rm otherwise.}
\end{cases}
\end{align}
Notice $v_x\geq 1$ in \eqref{v_pb00}, so \eqref{sup} is equal to
\be\label{sup2}
v_{xx}\int_{0}^{\infty} \frac{1}{2}\left(h^{*}\Big(z,-\frac{v_{xx}}{v_{x}}\Big)\right)^{2}\d F(z)
+v_{x}\int_{0}^{\infty} h^{*}\Big(z,-\frac{v_{xx}}{v_{x}}\Big) \d F(z).
\ee
For $y>0$, we have
\begin{align*}
\int_{0}^{\infty} \frac{1}{2}(h^{*}(z,y))^{2}\d F(z)
&=\int_0^{1/y}\frac{1}{2}z^2\d F(z)+\int_{1/y}^{\infty}\frac{1}{2}y^{-2}\d F(z)\\
&=-\frac{1}{2}z^{2}(1-F(z))\Big|_0^{1/y}+\int_0^{1/y}z(1-F(z))\d z+\frac{1}{2}y^{-2}(1-F(y^{-1}))\\
&=\int_0^{1/y}z(1-F(z))\d z,
\end{align*}
and
\begin{align*}
\int_{0}^{\infty} h^{*}(z,y)\d F(z)
&=\int_0^{1/y}z\d F(z)+\int_{1/y}^{\infty}y^{-1}\d F(z)\\
&=-z(1-F(z))\Big|_0^{1/y}+\int_0^{1/y}(1-F(z))\d z+y^{-1}(1-F(y^{-1}))\\
&=\int_0^{1/y}(1-F(z))\d z. 
\end{align*}
Similarly for $y\leq 0$, we have 
\begin{align*}
\int_{0}^{\infty} \frac{1}{2}(h^{*}(z,y))^{2}\d F(z)=\frac{1}{2}\int_{0}^{\infty} z^{2}\d F(z)=\frac{1}{2}\mu_2,
\end{align*}
and
\begin{align*}
\int_{0}^{\infty} h^{*}(z,y)\d F(z)=\int_{0}^{\infty} z\d F(z)=\mu_1. 
\end{align*}
Hence, \eqref{sup2} is equal to 
\bee
A\left(-\frac{v_{xx}}{v_{x}}\right) v_{xx}+B\left(-\frac{v_{xx}}{v_{x}}\right) v_x,
\eee
where the two functions $A(y)$ and $B(y)$ are defined by
\be\label{defA}
A(y)=\begin{cases}
\int_0^{1/y} z(1-F(z))\d z, & {\rm if}\;0< y<+\infty;\bigskip\\
\frac{1}{2}\mu_{2}=\int_0^{\infty} z(1-F(z))\d z, & {\rm if}\; y\leq 0,
\end{cases}
\ee
and
\be\label{defB}
B(y)=\begin{cases}
\int_0^{1/y}(1-F(z))\d z, & {\rm if}\;0< y<+\infty;\bigskip\\
\mu_{1}=\int_0^{\infty} (1-F(z))\d z, & {\rm if}\; y\leq 0.
\end{cases}
\ee 
Thanks to \eqref{distributionF}, it is not hard to show that both $A(y)$ and $B(y)$ are decreasing { Lipschitz continuous} functions, and satisfy
\begin{gather}\label{A_b}
0<A(y)\leq \frac{1}{2}\min\{y^{-2},\;\mu_2\}, \quad 0<B(y)\leq \min\{y^{-1},\;\mu_1\}, \quad {\rm for\;}y>0;\\
\label{Ap}
y^{3}A'(y)=y^{2}B'(y)=F(y^{-1})-1\leq 0, \quad {\rm for\; a.e.\;}{ y > 0}; \;\;\\
A'(y)=B'(y)=0,\quad {\rm for\;}{ y <0}. \label{AB}
\end{gather}
We now obtain
\be\label{calL2}
{\cal L}v=A\left(-\frac{v_{xx}}{v_{x}}\right) v_{xx}+B\left(-\frac{v_{xx}}{v_{x}}\right)v_x-\ga v_x-cv.
\ee
By this, we see problem \eqref{v_pb00}
is a variational inequality problem for a fully nonlinear elliptic operator subject to a gradient constraint. A usual approach to study this kind of problem is to transform it into a variational inequality problem for its gradient. Then the gradient constraint becomes value constraint and the new variational inequality becomes a well-studied obstacle problem; see \cite{DY09a, DY09b, GY19(2),HY16}. In this paper, we adopt the same approach.

Below we want to find (or more precisely, conjecture) a variational inequality for the gradient of $v$. Notice
\begin{align*}
&\quad\;\p_x\left[A\left(-\frac{v_{xx}}{v_{x}}\right) v_{xx}+B\left(-\frac{v_{xx}}{v_{x}}\right)v_x \right]\bigskip\\
&=A\left(-\frac{v_{xx}}{v_{x}}\right) v_{xxx}+B\left(-\frac{v_{xx}}{v_{x}}\right)v_{xx}+\left[A'\left(-\frac{v_{xx}}{v_{x}}\right)v_{xx}+B'\left(-\frac{v_{xx}}{v_{x}}\right)v_x\right]\p_x \left(-\frac{v_{xx}}{v_{x}}\right)\bigskip\\
&=A\left(-\frac{v_{xx}}{v_{x}}\right) v_{xxx}+B\left(-\frac{v_{xx}}{v_{x}}\right)v_{xx},
\end{align*}
where we used \eqref{Ap} and \eqref{AB} to get the last equation. By this, one can easily deduce that
\[
\p_x({\cal L}v)={\cal T}v_x,
\]
where the operator ${\cal T}$ is defined as
\be\label{calT}
{\cal T}u:=A\left(-\frac{u_x}{u}\right) u_{xx}+B\left(-\frac{u_x}{u}\right)u_x-\ga u_x-cu.
\ee

Next, we deduce a boundary condition for $v_x$. Define a continuous function
\be\label{f}
f(y):=-y A(y)+B(y)-\ga.
\ee
By \eqref{A_b}-\eqref{AB},
\begin{align}\label{fp}
f'(y)=-A(y)<0.
\end{align}
Hence $f$ is strictly decreasing. Also
\[f(0+)=\mu_1-\ga>0,\quad f(+\infty)=-\ga<0,\]
so $f$ has a unique root, which is positive, denote by $\la$ throughout this paper. See Figure \ref{fig0} for an illustration of the function $f$.
\begin{figure}[H]
\begin{center}
\begin{picture}(190,120)(0, 0)
\put(5,35){\vector(1,0){185}} \put(90,-5){\vector(0,1){135}}
\qbezier(30, 115)(60, 85)(90, 55)
\qbezier(90, 55)(120, 30)(180, 25)
\put(90, 55){\circle*{3}}
\put(56, 51){\footnotesize $\mu_{1}-\ga$}
\qbezier[100](10, 20)(105, 20)(185, 20)
\put(72, 12){\footnotesize $-\ga$}
\put(90, 20){\circle*{3}}
\put(122, 25){\footnotesize $\la$}
\put(128, 35){\circle*{3}}
\end{picture}
\caption{The function $f$ is continuous and strictly deceasing, and admits a unique positive root $\la$.} \label{fig0}\end{center}
\end{figure}
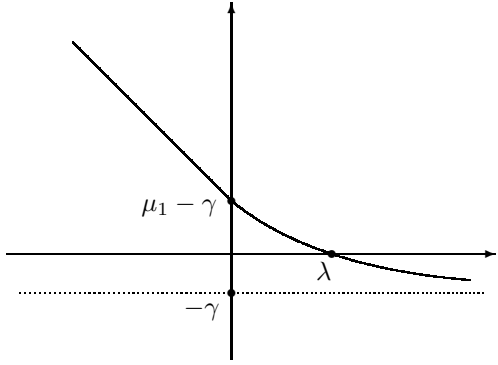
On the other hand, owing to the boundary condition $v(0,t)=0$ and that $v_t-{\cal L}v=0$ near $\{x=0\}$, we should have
\be\label{Lv0}
\left(A\left(-\frac{v_{xx}}{v_{x}}\right) v_{xx}+B\left(-\frac{v_{xx}}{v_{x}}\right)v_x-\ga v_x \right)\Bigg|_{x=0}=0.
\ee
Dividing by $v_x$, it yields $f\left(-\frac{v_{xx}}{v_{x}}\right)\Big|_{x=0}=0$. So $-\frac{v_{xx}}{v_{x}}\big|_{x=0}=\la$ by the strictly monotonicity of $f$. This leads to a boundary condition
\be\label{LC}
\Big(\la v_x+v_{xx}\Big)\Big|_{x=0}=0.
\ee

Combining the above results, we suggest the following variational inequality for $u=v_x$
\be\label{u_pb}
\begin{cases}
\min\{u_t-{\cal T}u,\; u-1\}=0 \quad
{\rm in}\quad \Q_T,\bigskip\\
\big(\la u+u_x\big)(0,t)=0, \quad 0<t\leq T,\bigskip\\
u(x,0)=1, \quad x>0.
\end{cases}\ee
This is an obstacle problem for a quasilinear elliptic operator with mixed boundary conditions. We study it by the penalty approximation method; see \cite{DY09b, GY19(2),Yi08}. We first prove the existence of a solution to problem \eqref{u_pb}, and then construct a solution to problem \eqref{v_pb00} from it. In the last part of this section, we show the uniqueness, which is indeed not necessary as the verification theorem also guarantees the uniqueness.

For each sufficiently small $\ep>0$, let $\beta _\ep(\cdot)$ be a penalty function satisfying
\begin{eqnarray*}
&&\beta _{\ep}(\cdot)\in C^{2}(-\infty,+\infty), \quad\beta _{\ep}(0)=-c, \quad\beta _{\ep}(x)=0\; \text{ for }\; x\geq \ep, \bigskip\\
&& \beta _{\ep
}(\cdot)\leq 0, \quad \beta _{\ep}^{\prime}(\cdot)\geq 0,\quad \beta _{\ep
}^{\prime \prime}(\cdot)\leq 0, \quad\lim\limits_{\ep \rightarrow 0+} \beta _{\ep}(x)=\begin{cases}
0, &{\rm if}\;\; x>0,\vspace{2mm} \\
-\infty, &{\rm if}\;\; x<0.
\end{cases}\end{eqnarray*}%
See Figure \ref{fig_beta} for an illustration of the penalty function $\beta_\ep$.
\begin{figure}[H]
\begin{center}
\begin{picture}(230,120)
\put(45,80){\vector(1,0){150}}
\put(110,-5){\vector(0,1){120}}
\qbezier(104,0)(115,80)(142,80)
\put(180,70){$x$}
\put(100,105){$y$}
\put(138,70){$\ep$}
\put(90,31){$-c$}
\put(110,33){\circle*{3}}
\put(141,80){\circle*{3}}
\end{picture} \caption{The penalty function $\beta_\ep$.} \label{fig_beta}
\end{center}
\end{figure}
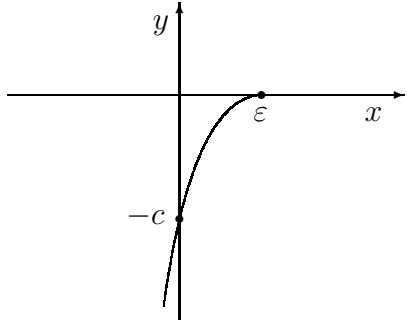
Because the left boundary condition and initial condition in \eqref{u_pb} are not consistent at $(0,0)$, in order to show the existence of a solution, we choose $f_\ep\in C^2([0,+\infty))$ to satisfy
\bee
f_{\ep}(t)=
\begin{cases}
\la,& t=0,\\
{\rm decreasing},& 0\leq t<\ep,\\
0,& t\geq\ep,
\end{cases}\eee
and consider the following penalty approximation problem,
\be\label{ue_pb}
\begin{cases}
u^\ep_t-{\cal T} u^\ep+\beta _\ep(u^\ep-1)=0 \quad{\rm in}\quad \QLT:=(0,L)\times(0,T],\bigskip\\
\big(\la u^\ep+u^\ep_x\big)(0,t)=f_\ep(t),\quad 0<t\leq T,\bigskip\\
u^\ep(L,t)=1,\quad 0<t\leq T,\bigskip\\
u^\ep(x,0)=1,\quad 0<x<L.
\end{cases}\ee
where $L$ is any fixed large positive constant and ${\cal T}$ is defined by \eqref{calT}.

The following result is useful for us to handle the above problem.
\begin{lemma}\label{G}
Let $A$ and $B$ be defined by \eqref{defA} and \eqref{defB}, respectively, and define
\[
F(y,z):=A\left(-\frac{y}{z}\right),\qquad G(y,z):=B\left(-\frac{y}{z}\right).
\]
Then $F$ and $G$ are uniformly Lipschitz continuous in $(-\infty,\infty)\times[1,\infty) $.
\end{lemma}

\begin{proof}
As $A(y)$ is a Lipschitz continuous function, $F$ is continuous in $(-\infty,\infty)\times[1,\infty)$. For any $z\geq 1$,
\[
\p_y F(y,z)=A'\left(-\frac{y}{z}\right)\frac{-1}{z}=
\begin{cases}
\left(-\frac{z}{y}\right)^{3}\left(1-F\left(-\frac{z}{y}\right)\right)\frac{1}{z}, & {\rm if}\quad y<0;\medskip\\
0,& {\rm if}\quad y> 0,
\end{cases}\]
and
\[
\p_z F(y,z)=A'\left(-\frac{y}{z}\right)\frac{y}{z^2}=
\begin{cases}
\left(-\frac{z}{y}\right)^{2}\left(1-F\left(-\frac{z}{y}\right)\right)\frac{1}{z}, & {\rm if}\quad y<0;\medskip\\
0,& {\rm if}\quad y> 0.
\end{cases}
\]
Thanks to \eqref{distributionF} and $z\geq 1$, we see $\p_y F$ and $\p_z F$ are bounded, so $F$ is uniformly Lipschitz continuous in $(-\infty,\infty)\times[1,\infty) $. In the same way, we can show $G$ is also uniformly Lipschitz continuous in $(-\infty,\infty)\times[1,\infty) $.
\end{proof}

We first establish a comparison principle for problem \eqref{ue_pb}.
\begin{lemma}\label{lem:com}
Suppose $u_1,\;u_2\in C^{2,1}\big(\QLT\big)\bigcap C\big(\overline{\QLT}\big)$ satisfy
\bee
\begin{cases}
\p_t u_1-{\cal T} u_1+\beta_\ep(u_1-1)
\leq \p_t u_2-{\cal T} u_2+\beta_\ep(u_2-1) \quad{\rm in}\quad \QLT,\bigskip\\
\big(\la u_1+\p_x u_1\big)(0,t)\geq \big(\la u_2+\p_x u_2\big)(0,t),\quad 0<t\leq T,\bigskip\\
u_1(L,t)\leq u_2(L,t),\quad 0<t\leq T,\bigskip\\
u_1(x,0)\leq u_2(x,0),\quad 0<x<L.
\end{cases}\eee
If $u_1,\;u_2\geq1$, then
\be\label{com}
u_1(x,t)\leq u_2(x,t),\quad (x,t)\in \QLT.
\ee
\end{lemma}

\begin{proof} Let
\[
w_1(x,t)=e^{x/\la}u_1(x,t),\quad w_2(x,t)=e^{x/\la}u_2(x,t).
\]
Then
\[
\begin{cases}
\p_t w_1-e^{x/\la}{\cal T} (e^{-x/\la}w_1)+e^{x/\la}\beta_\ep(e^{-x/\la}w_1-1)
\leq \p_t w_2-e^{x/\la}{\cal T} (e^{-x/\la}w_2)+e^{x/\la}\beta_\ep(e^{-x/\la}w_2-1),\bigskip\\
\p_x w_1(0,t)\geq \p_x w_2(0,t),\quad 0<t\leq T,\bigskip\\
w_1(L,t)\leq w_2(L,t),\quad 0<t\leq T,\bigskip\\
w_1(x,0)\leq w_2(x,0),\quad 0<x<L.
\end{cases}\]
By Lemma \ref{G}, the assumption that $u_1$, $u_2\geq 1$ is sufficient to guarantee that $F(\cdot,\cdot)$ and $G(\cdot,\cdot)$ in ${\cal T} (e^{-x/\la}w_i)$ are Lipschitz continuous on $w_i$, $w_{ix}$, $i=1,2$. By the comparison principle for nonlinear equations (see \cite{Li96} Theorem 14.3), we obtain $w_1\leq w_2$ in $\QLT. $
\end{proof}

\begin{lemma}\label{thm:ue}
There exists a solution $u^\ep \in C^{2,1}\big(\overline{\QLT}\big)$ to Problem (\ref{ue_pb}). Moreover, for sufficiently small $\ep>0$, $u^\ep $ satisfies
\be\label{ue}
1\leq u^\ep \leq \frac{Ke^{\La t}}{x+1/\la}, \quad{\rm in}\quad \QLT,
\ee
where
\bee
K=L+1/\la+1,\quad
\La=\frac{\mu_{2}}{\ga^2}+\ga \la>0.
\eee
\end{lemma}

\begin{proof} Using the Leray-Schauder fixed point theorem (see \cite{Ev16} Theorem 4 on page 541) and the embedding theorem (see \cite{Li96} Theorem 6.8), we get the existence of a $C^{1+\al,\frac{1+\al}{2}}\big(\overline{\QLT}\big) (0<\al<1)$ solution $u^\ep$ to problem \eqref{ue_pb}. By the Schauder estimation (see \cite{Li96} Theorem 4.23), we also have $u^\ep \in C^{2+\al,1+\frac{\al}{2}}\big(\overline{\QLT}\big)$.

It is left to establish (\ref{ue}). Let $\phi\equiv 1$, then
\[
\begin{cases}
\phi_t-A\left(-\frac{u^\ep_x}{u^\ep}\right) \phi_{xx}-B\left(-\frac{u^\ep_x}{u^\ep}\right)\phi_x+\ga \phi_x+c\phi+\beta_\ep(\phi-1)=0,\bigskip\\
\big(\la\phi+\phi_{x}\big)(0,t)=\la\geq f_\ep(t),\quad 0<t\leq T,\bigskip\\
\phi(L,t)=1,\quad 0<t\leq T,\bigskip\\
\phi(x,0)=1,\quad 0<x<L,
\end{cases}\]
Regarding $A\left(-\frac{u^\ep_x}{u^\ep}\right)$ and $B\left(-\frac{u^\ep_x}{u^\ep}\right)$ as known coefficients, we can establish a comparison principle for the above PDE (which is similar to Lemma \ref{lem:com}) which will lead to $u^\ep \geq \phi=1$. 

Let $\Phi=Ke^{\La t}/(x+1/\la)$. Then, in $ \QLT$, $\Phi\geq 1+\ep$ for sufficiently small $\ep>0$, so we have $\beta_\ep(\Phi-1)=0$. Notice that
\bee
\Phi_t=\La \Phi,\quad
\Phi_x=-\frac{\Phi}{x+1/\la}<0,\quad \Phi_{xx}=\frac{2\Phi}{(x+1/\la)^{2}}>0,
\eee
so, by \eqref{A_b},
\begin{multline*}
\Phi_t-A\left(-\frac{\Phi}{\Phi_x}\right) \Phi_{xx}-B\left(-\frac{\Phi}{\Phi_x}\right)\Phi_x+\ga \Phi_x+c\Phi+\beta_\ep(\Phi-1)\\
\geq \Phi_t-\frac{1}{2}\mu_{2} \Phi_{xx}+\ga \Phi_x+c\Phi=\left(\La-\frac{\mu_2}{(x+1/\la)^2}-\frac{\ga}{x+1/\la}+c\right)\Phi\;\geq\; 0.
\end{multline*}
Together with boundary conditions
\bee
\begin{cases}
\big(\la\Phi+\Phi_x\big)(0,t)=0\leq f_\ep(t), & 0<t\leq T, \bigskip\\
\Phi(L,t)\geq1, & 0<t\leq T, \bigskip\\
\Phi(x,0)\geq1, & 0<x\leq L,
\end{cases}\eee
applying Lemma \ref{lem:com}, we obtain $u^\ep \leq \Phi$.
\end{proof}

Before passing to the limit, we show some properties of $u^\ep$.
\begin{lemma}\label{prop:utux}
We have
\be
&&u^\ep _t\geq 0, \label{uet}\bigskip\\
&&u^\ep _x\leq 0, \label{uex}
\ee
\end{lemma}
\begin{proof} We first prove \eqref{uet}. For any $0<\Delta <T$, let $\widetilde{u}^\ep (x,t)=u^\ep (x,t+\Delta)$, then both $\widetilde{u}^\ep $ and $u^\ep $ satisfy the equation in \eqref{ue_pb} in the domain $(0,L)\times(0,T-\Delta ]$. Moreover,
\bee
\begin{cases}
\Big(\la\widetilde{u}^\ep+\widetilde{u}^\ep _x\Big)(0,t)=f_\ep(t+\Dt)\leq f_\ep(t)=\Big(\la u^\ep+u^\ep _x\Big)(0,t) \quad{\rm in}\quad & (0,L)\times(0,T-\Delta ], \bigskip\\
\widetilde{u}^\ep (L,t)=u^\ep (L,t)=1, & 0<t\leq T-\Delta, \bigskip\\
\widetilde{u}^\ep (x,0)=u^\ep (x,\Delta)\geq1=u^\ep (x,0), & 0<x< L.
\end{cases}\eee
Applying Lemma \ref{lem:com} we have $\widetilde{u}^\ep \geq u^\ep $ in $(0,L)\times(0,T-\Delta ]$, which implies \eqref{uet}.

To prove \eqref{uex}, we differentiate the equation in \eqref{ue_pb} w.r.t. $x$ and obtain
\begin{multline}\label{uex_eq}
\p_t u^\ep _x-\p_x\left[A\left(-\frac{u^\ep _{x}}{u^\ep}\right)\p_xu^\ep _{x}\right]-B\left(-\frac{u^\ep _{x}}{u^\ep}\right) \p_xu^\ep _{x}\bigskip\\
-B'\left(-\frac{u^\ep _{x}}{u^\ep}\right)\p_x\left(-\frac{u^\ep _{x}}{u^\ep}\right) u^\ep _{x}
+cu^\ep _x+\beta'_{\ep}(u^\ep-1)u^\ep _x=0.
\end{multline}
Note that
\[
\p_x\left(-\frac{u^\ep _x}{u^\ep}\right)=-\frac{u^\ep _{xx}}{u^\ep}+\left(-\frac{u^\ep _x}{u^\ep}\right)^{2},
\]
so \eqref{uex_eq} can be written as
\begin{multline}\label{uex_eq1}
\p_t u^\ep _x-\p_x\left[\left\{A\left(-\frac{u^\ep _{x}}{u^\ep}\right)\right\}\p_xu^\ep _{x}\right]-\left\{B\left(-\frac{u^\ep _{x}}{u^\ep}\right)+B'\left(-\frac{u^\ep _{x}}{u^\ep}\right)\frac{-u^\ep _{x}}{u^\ep}\right\}\p_xu^\ep _{x}\bigskip\\
+\left\{-B'\left(-\frac{u^\ep _{x}}{u^\ep}\right)\left(-\frac{u^\ep _x}{u^\ep}\right)^{2}+c+\beta'_{\ep}(u^\ep-1)\right\} u^\ep _x=0.
\end{multline}
It is a linear equation for $u^\ep _x$ in divergence form if we regard the terms in $\{\cdots\}$ as known coefficients. By \eqref{A_b}-\eqref{AB},
\begin{gather*}
\left|A\left(y\right)\right|\leq \mu_2,\quad
\left|B\left(y\right)\right|\leq \mu_1,\quad
\left|B'\left(y\right)y\right|\leq |y^{-1} (1-F(y^{-1}))|\leq \mu_1,\\
\left|B'\left(y\right)y^{2} \right|\leq | 1-F(y^{-1})|\leq 1,\quad
\beta'_{\ep}(u^\ep-1)\geq 0,
\end{gather*}
so all the coefficients of \eqref{uex_eq1} are bounded, except the last one which is bounded from below. 
Moreover, since
\bee
u^\ep _x(0,t)=(f_\ep(t)-\la u^\ep (0,t)),
\eee
together with $f_\ep\leq \la\leq \la u^\ep$, it implies that
\bee
u^\ep _x(0,t)\leq 0.
\eee
From $u^\ep \geq1$ and $u^\ep (L,t)=1$ we have
\bee
u^\ep _x(L,t)\leq0.
\eee
Moreover,
\bee
u^\ep _x(x,0)=0,
\eee
by the maximum principle for weak solution (see \cite{Li96} Corollary 6.16), we deduce that $u^\ep _x\leq0$.
\end{proof}

\begin{lemma}\label{prop:ue_mix}
We have
\be\label{ue_mix}
\la u^\ep+u^\ep _x\geq 0.
\ee
\end{lemma}
\begin{proof}
If $u^\ep(0,t)=1$, since $u^\ep\geq 1$ and $u^\ep_{x}\leq 0$, we have $u^\ep(x,t)\equiv1$ and $u^\ep_{x}(x,t)\equiv0$ for all $x\geq 0$, so $(\la u^\ep+u^\ep _x)(x,t)\equiv \la> 0$, and \eqref{ue_mix} follows.

Otherwise $u^\ep(0,t)>1$ which leads to
\be\label{ue0}
(c u^\ep+\beta_\ep(u^\ep-1))(0,t)>c+\beta_\ep(0)=0.
\ee
Consider the equation in \eqref{ue_pb}, i.e.
\[
u^\ep _t-A\left(-\frac{u^\ep _x}{u^\ep}\right)u^\ep _{xx}-B\left(-\frac{u^\ep _x}{u^\ep}\right)u^\ep _x+\ga u^\ep _x+c u^\ep+\beta_\ep(u^\ep-1)=0.
\]
If $u^\ep _x=0$, then it reaches its global minimum value 0, so $u^\ep _{xx}=0$. Together with $u^\ep _t\geq0$, the above equation gives $c u^\ep+\beta_\ep(u^\ep-1)\leq 0$, which by the definition of $\beta_\ep$ implies $u^\ep=1$. Therefore, $\la u^\ep+u^\ep _x=\la>0$, and \eqref{ue_mix} follows.
Otherwise, we have $u^\ep _x<0$, dividing both sides by $u^\ep_{x}$ and using the identity
\[
-\frac{u^\ep_{xx}}{u^\ep_{x}}=\left[\p_x\left(-\frac{u^\ep}{u^\ep_{x}}\right)+1\right]\left(\frac{-u^\ep _x}{u^\ep}\right),
\]
it follows
\bee
\frac{u^\ep_{t}}{u^\ep_{x}}+A\left(-\frac{u^\ep_{x}}{u^\ep}\right) \left[\p_x\left(-\frac{u^\ep}{u^\ep_{x}}\right)+1\right]\left(\frac{-u^\ep _x}{u^\ep}\right)-B\left(-\frac{u^\ep_{x}}{u^\ep}\right)+\ga+\frac{cu^\ep+\beta_\ep(u^\ep-1)}{u^\ep_{x}}=0.
\eee
Denote $\nu^\ep=-u^\ep/u^\ep _x$. Then
\bee
A \left(\frac{1}{\nu^\ep}\right) \frac{\nu^\ep_{x}+1}{\nu^\ep}-B\left(\frac{1}{\nu^\ep}\right)+\ga-\frac{cu^\ep+\beta_\ep(u^\ep-1)}{u^\ep}\nu^\ep=-\frac{u^\ep_{t}}{u^\ep_{x}}\geq 0
\eee
by \lemref{prop:utux}. We get
\[
\nu^\ep_x\geq \frac{1}{A \left(\frac{1}{\nu^\ep}\right)}\left[f \left(\frac{1}{\nu^\ep}\right)\nu^\ep+\frac{c u^\ep+\beta_\ep(u^\ep-1)}{u^\ep}(\nu^\ep)^2\right],
\]
where $f(\cdot)$ is defined in \eqref{f}. Notice that
\[
c u^\ep+\beta_\ep(u^\ep-1)\geq c+\beta_\ep(0)=0,
\]
and $f(\frac{1}{z})\geq 0$ when $z\geq 1/\la$, so
\be\label{nu3}
\nu^\ep_x\geq 0\quad{\rm if}\quad \nu^\ep\geq 1/\la.
\ee
Moreover, by $f(\la)=0$, \eqref{ue0} and
\be\label{nu1}
\nu^\ep(0,t)=1/\la,
\ee
it yields
\be\label{nu2}
\nu^\ep_x(0,t)>0.
\ee
Combining with \eqref{nu1}, \eqref{nu2} and \eqref{nu3} we get $\nu^\ep\geq 1/\la$, which implies \eqref{ue_mix}.

\end{proof}

By \lemref{thm:ue}, \lemref{prop:utux} and \lemref{prop:ue_mix}, we see that $|u^\ep _x|\leq \la u^\ep\leq K \la^{2} e^{\La T}$ in $\QLT$. This provides an upper bound for $|u^\ep _x|$, independent of $\ep$, so $u^\ep$ are uniformly Lipschitz continuous in $x$. Moreover, by \lemref{prop:utux} and \lemref{prop:ue_mix} and the monotonicity of $A$,
\be\label{Aga}
A\left(-\frac{u^\ep _x}{u^\ep}\right) \geq A(\la)>0.
\ee
This confirms the uniform parabolic condition in \eqref{ue_pb} for the equation of $u^\ep $ (where we regard $A\left(-u^\ep _x/u^\ep\right)$ and $B\left(-u^\ep _x/u^\ep\right)$ as known coefficients in the operator ${\cal T}$).

\begin{lemma}\label{prop:uexx}
We have
\be\label{uexx}
u^\ep_{xx}\geq 0.
\ee
\end{lemma}

\begin{proof} By the equation in \eqref{ue_pb}, \lemref{thm:ue}, \lemref{prop:utux}, $f(\la)=0$, and \eqref{Aga}, we have
\begin{align*}
A\left(-\frac{u^\ep _x}{u^\ep}\right)u^\ep _{xx}
&=u^\ep _t-B\left(-\frac{u^\ep _x}{u^\ep}\right)u^\ep _x+\ga u^\ep _x+c u^\ep+\beta_\ep(u^\ep-1)\\
&\geq \left[-B\left(-\frac{u^\ep _x}{u^\ep}\right)+\ga\right]u^\ep _x
\geq \left[-B(\la)+\ga\right]u^\ep _x
=-\la A(\la) u^\ep _x
\geq 0.
\end{align*}
This gives \eqref{uexx}.
\end{proof}

Now we give some uniform norm estimates for $u^\ep$. First, we rewrite the equation of $u^\ep$ in the divergence form. By \eqref{ue_pb}, \eqref{Ap} and \eqref{AB} we have
\[
u^\ep _t-\p_x\left( A\left(-\frac{u^\ep _x}{u^\ep}\right)u^\ep _x-B\left(-\frac{u^\ep _x}{u^\ep}\right)u^\ep \right)+\ga u^\ep _x+c u^\ep+\beta_\ep(u^\ep-1)=0.
\]
Since $A(\cdot),\;B(\cdot)$ are bounded, applying $C^{\al,\frac{\al}{2}}$ estimate (see \cite{Li96} Theorem 6.33 for the interior estimate and Theorem 6.33 for the boundary estimate) we have
\[
|u^\ep|_{\al,\QLT}\leq C(|u^\ep|_{0,\QLT}+|\beta_\ep(\cdot)|_{L_p(\QLT)}+1)\leq C.
\]
Rewrite \eqref{uex_eq1} as
\begin{multline*}
\p_t u^\ep _x-\p_x\left[A\left(-\frac{u^\ep _{x}}{u^\ep}\right)\p_xu^\ep _{x}\right]-\left\{B\left(-\frac{u^\ep _{x}}{u^\ep}\right)+B'\left(-\frac{u^\ep _{x}}{u^\ep}\right)\frac{-u^\ep _{x}}{u^\ep}\right\}\p_xu^\ep _{x}\bigskip\\
+\left\{-B'\left(-\frac{u^\ep _{x}}{u^\ep}\right)\left(-\frac{u^\ep _x}{u^\ep}\right)^{2}+c \right\} u^\ep _x=-\p_x\Bigg[\beta_{\ep}(u^\ep-1)\Bigg].
\end{multline*}
Applying $C^{\al,\frac{\al}{2}}$ interior (with partial boundary) estimate, we obtain
\[
|u^\ep_x|_{\al,\Q^{r/2}}\leq C(|u_x^\ep|_{0,\Q^{r/4}}+|\beta_\ep(\cdot)|_{L_p(\Q^{r/4})}+1)\leq C.
\]
where
\[
\Q^r:=\QLT\setminus \{(x,t)\;|\;x^2+t^2\leq r\}.
\]
According to Lemma \ref{G}, $A\left(-u^\ep _x/u^\ep\right)$ and $B\left(-u^\ep _x/u^\ep\right)$ are uniform $C^{\al,\frac{\al}{2}}$ in $\Q^{r/2}$. So we can apply $W^{2,1}_p$ interior estimate (see \cite{Li96} Theorem 7.13 for the interior estimate and Theorem 7.17 for the boundary estimate) to \eqref{ue_pb} in $\QLT$ to obtain
\[
|u^\ep|_{W^{2,1}_p(\Q^{r})}\leq C(|u^\ep|_{L_p(\Q^{r/2})}+|\beta_\ep(u^\ep-1)|_{L_p(\Q^{r/2})}+1)\leq C.
\]
We emphasize that $C$s in the above estimates are independent of $\ep$, so there exits $u\in W^{2,1}_{p,\;{\rm loc}}(\QLT)\bigcap C(\overline{\QLT})$ and a subsequence of $u^\ep$ (still denoted by $u^\ep$) such that
\[
u^\ep\longrightarrow u \quad\hbox{weakly in}\quad W^{2,1}_p(\Q^{r})\quad\hbox{and uniformly in}\quad C(\overline{\QLT}).
\]
Then, $u$ is a solution to problem \eqref{u_pb} in $W^{2,1}_{p,\rm loc}\big(\Q_T\big)\bigcap C\big(\overline{\Q_T}\big)$.

Now, set
\[
v(x,t)=\int_0^xu(y,t)\d y,
\]
we come to prove $v$ is a solution to problem \eqref{v_pb00} in $\QLT$. The initial and boundary conditions are clearly satisfied. Owing to $v(0,t)=0$ and $v_t(0,t)=0$, together with the boundary condition in \eqref{u_pb} that $(\la u+u_x)(0,t)=0$, we have $(v_t-{\cal L} v)(0,t)=0$. Therefore,
\be\label{uv}
(v_t-{\cal L} v)(x,t)=(v_t-{\cal L} v)(0,t)+\int_0^x\p_x(v_t-{\cal L} v)(y,t)\d t=\int_0^x(u_t-{\cal T} u)(y,t)\d t \geq 0. \quad
\ee
On the other hand, if $v_x(x,t)=u(x,t)> 1$, then, by \eqref{uex}, $u(y,t)> 1$ for all $y\in[0,x]$, which implies $(u_t-{\cal T} u)(y,t)=0$ for $y\in[0,x]$, thus the inequality in \eqref{uv} becomes equality. Hence $v$ satisfies the variational inequality in problem \eqref{v_pb00} in $\QLT$.

Moreover, the estimates \eqref{vx}-\eqref{vga} follow from \eqref{ue}, \eqref{uex}, \eqref{uexx}, \eqref{uet} and \eqref{ue_mix}.

We come to ascertain the order of smoothness of $v$. Now, we already proved $v_x=u\in W^{2,1}_{p}(Q^r)\bigcap C\big(\overline{\QLT}\big)$. The Sobolev embedding theorem implies that $v_{xx}=u_x\in C\big(\overline{\Q^r}\big)$. Moreover, using the method in \cite{Fr75}, we can prove $v_{xt}=u_t$ is continuous passing through the free boundary.

Next, we prove the uniqueness. Suppose $v_1$, $v_2$ are two solutions to \eqref{v_pb00}. Set ${\cal N}=\{\p_xv_1>\p_xv_2\}$, then
\[
\begin{cases}
\p_tv_1-{\cal L}v_1=0,\quad \p_tv_2-{\cal L}v_2\geq0,\quad (x,t)\in {\cal N},\\
v_1=v_2=0,\quad (x,t)\in\p{\cal N}\cap \{x=0\},\\
\p_xv_1=\p_xv_2,\quad (x,t)\in\p{\cal N}\setminus (\{x=0\}\cup\{t=0\}\cup\{t=T\}),\\
v_1=v_2=x,\quad (x,t)\in\p{\cal N}\cap \{t=0\}.
\end{cases}\]
Apply the comparison principle for fully nonlinear equation (see \cite{Li96} Theorem 14.3) we have $v_2\geq v_1$ in ${\cal N}$, which implies
\[
\{\p_xv_1> \p_xv_2\}\subset\{v_2\geq v_1\},
\]
i.e.
\[
{\cal C}:=\{v_2< v_1\}\subset\{\p_xv_1\leq \p_xv_2\}.
\]
If ${\cal C}$ is nonempty, using the fact that $v_2=v_1$ on the left boundary of $\cal C$ and $\p_xv_1\leq \p_xv_2$ in ${\cal C}$, we get $v_1\leq v_2$ in ${\cal C}$, which is impossible. This completes the proof of the uniqueness.

Let $x_2$ be defined in \eqref{x2} and choose $L> x_2$. Using a similar argument in Section \ref{sec:D} leads to $v_x(x,t)=1$ for $x\in [x_2,L]$, so we can extend our solution to the unbounded domain $\QT$ by setting $v(x,t)=v(L,t)+(x-L)$ for $x>L$. Then after extension, $v\in C^{2,1}(\overline{\QT}\setminus\{(0,0)\})\bigcap C(\overline{\QT})$ is a unique solution to \eqref{v_pb00} in $\QT$. Moreover, the properties \eqref{vx}-\eqref{vga} remain true in $\QT$.

Furthermore, Lemma \ref{G} and the equation in $\{v_{x}>1\}$ implies $v_{xt}=u_t\in C\big(\{v_{x}>1\}\setminus\{(0,0)\}\big)$, so $v_{xt}\in C\big(\overline{\QT}\setminus\{(0,0)\}\big)$. Hence we have
\[
v,\;v_x\in C\big(\overline{\QT}\big),\quad v_{xx},\;v_t,\;v_{xt}\in C\big(\overline{\QT}\setminus\{(0,0)\}\big).
\]
This completes the proof.

\section{Proof of \thmref{thm:averi}.}\label{sec:verify}
\setcounter{equation}{0}

Denote $\WV(x,t)=v(x,T-t)$.
We first prove $\WV(x,t)\geq V(x,t)$. For any admissible retained loss policy $\cH^t=\{H_s(z)\}_{s\geq t}$ and a dividend-payout policy $\cL^t=\{L_s\}_{s\geq t}$, assume $R_s$ is the solution to \eqref{Rt} with the control pair $(\cH^t, \cL^t)$, and $\deathtime$ is the ruin time of $R_s$ defined by \eqref{tau}. Then by It\^o's formula,
\begin{align*}
\WV(x,t)&=\E\Bigg[e^{-c(T\wedge\deathtime-t)}\WV(R_{T\wedge\deathtime},T\wedge\deathtime)\Bigg]\\
&\quad\;+\E\Bigg[\int_{t}^{T\wedge \deathtime}e^{-c(s-t)}\Bigg(-\WV_t-\frac{\WV_{xx}}{2}\int_0^\infty H_s(z)^2 \d F(z)\\
&\quad\;\quad\quad\quad\quad\quad\quad\quad\quad\quad\quad\quad\quad
-\WV_x\int_0^\infty H_s(z) \d F(z)+\ga \WV_x+c\WV\Bigg)(R_{s-},s)\d s\Bigg]\\
&\quad\;+\E \Bigg[ \int_t^{T\wedge \deathtime}e^{-c(s-t)}\WV_x(R_{s-},s)\d L^c_s \Bigg]-\E \sum\limits_{t\leq s\leq T\wedge\deathtime}e^{-c(s-t)}(\WV(R_s,s)-\WV(R_{s-},s)),
\end{align*}
where $L^c_s$ is the continuous part of $L_s$. The first two expectations are non-negative since $\WV\geq0$ and $-\WV_t-\LL \WV\geq 0$ by \eqref{v_pb00}. Meanwhile, since $\WV_x\geq1$ and $R_{s}\leq R_{s-}$, we have
\begin{align*}
\E\Bigg[\int_t^{T\wedge \deathtime}e^{-c(s-t)}\WV_x(R_{s-},s)\d L^c_s \Bigg]\geq \E\Bigg[\int_t^{T\wedge \deathtime}e^{-c(s-t)}\d L^c_s \Bigg],
\end{align*}
and
\begin{align*}
\WV(R_s,s)-\WV(R_{s-},s)\leq R_s-R_{s-}=L_{s-}-L_s.
\end{align*}
Thus
\[
\WV(x,t)\geq
\E\Bigg[\int_t^{T\wedge \deathtime}e^{-c(s-t)}\d L^c_s+\sum\limits_{t\leq s\leq T\wedge\deathtime}e^{-c(s-t)}(L_s-L_{s-})\Bigg]=\E\Bigg[\int_{t-}^{T\wedge \deathtime}e^{-c(s-t)}\d L_s\Bigg].
\]
Since the policies are arbitrary chosen, it implies $\WV(x,t)\geq V(x,t)$.

We now prove the reverse inequality $\WV(x,t)\leq V(x,t)$. Let $d^*$ be defined as in the statement. Then
\[d^{*}(s)=\inf\{x\geq0\;|\;\WV_x(x,T-s)=1\},\quad s\in[0,T].\]
Because $\WV$ is concave in $x$ by hypothesis, it yields
\begin{align*}
\WV_x(x,s)>1,\quad {\rm if }\quad x< d^{*}(T-s).
\end{align*}
By the equation in \eqref{v_pb00}, it follows
\begin{align}\label{wvx}
-\WV_t-{\cal L}\WV=0,\quad {\rm if }\quad x\leq d^{*}(T-s),
\end{align}
where the $C^{2,1}$ continuity of $\WV$ ensures the above equation holds at $(d^{*}(T-s),s)$. Let $I^*$ and $L^*$ be defined as in the statement. Then
\[H_{s}^{*}(z, R^{*}_{s-})=z-I_{s}^{*}(z, R^{*}_{s-})=h^{*}\left(z,-\frac{\WV_{xx}(R^{*}_{s-},s)}{\WV_{x}(R^{*}_{s-},s)}\right),\]
where $h^{*}$ is defined by \eqref{h*}. With the above feedback controls $H^*$ and $L^*$, using the property that $R^*_s\leq d^{*}(T-s)$ for $s>t$, one can show \eqref{Rt} admits a strong solution $R^*_s$. Moreover, $R^*_s$ is continuous for $s>t$ and it is inside the closure of the waiting region all the time. 
Denote $\deathtime^*$ be the corresponding ruin time.

Now we show that the controls defined above are indeed the optimal controls. By It\^o's formula,
\begin{align}\label{itofor}
\WV(x,t) &=\E\Bigg[e^{-c(T\wedge\deathtime^*-t)}\WV(R^*_{T\wedge\deathtime^*},T\wedge\deathtime^*)\Bigg]\nonumber\\
&\quad\;+\E\Bigg[\int_t^{T\wedge \deathtime^*}e^{-c(s-t)}\Bigg(-\WV_t-\frac{\WV_{xx}}{2}\int_0^\infty H^*_s(z, R^{*}_{s-})^2 \d F(z)\nonumber\\
&\quad \quad\quad\quad\quad\quad\quad\quad\quad\quad\quad\quad\quad
-\WV_x\int_0^\infty H^*_s(z, R^{*}_{s-}) \d F(z)+\ga \WV_x+c\WV\Bigg)(R^*_{s-},s)\d s\Bigg]\nonumber\\
&\quad\;+\E \Bigg[ \int_t^{T\wedge \deathtime^*}e^{-c(s-t)}\WV_x(R^*_{s-},s)\d L^{*c}_s \Bigg]-\E \Bigg[ \sum\limits_{t\leq s\leq T\wedge\deathtime^*}e^{-c(s-t)}(\WV(R^*_s,s)-\WV(R^*_{s-},s))\Bigg].
\end{align}

If $\deathtime^{*}<T$, then $R^*_{T\wedge\deathtime^*}=R^*_{\deathtime^*}=0$, so $\WV(R^*_{T\wedge\deathtime^*},T\wedge\deathtime^*)
=\WV(0, \deathtime^*)=0$ by the boundary condition in \eqref{v_pb00}. Otherwise $\WV(R^*_{T\wedge\deathtime^*},T\wedge\deathtime^*)
=\WV(R^*_{T},T)=0$, again by the boundary condition in \eqref{v_pb00}. Therefore, the first expectation in \eqref{itofor} is zero.

By our choice of $H^{*}$, we see that at $(R^*_{s-},s)$
\begin{align*}
&\quad\;-\WV_t-\frac{\WV_{xx}}{2}\int_0^\infty H^*_s(z, R^{*}_{s-})^2 \d F(z) -\WV_x\int_0^\infty H^*_s(z, R^{*}_{s-}) \d F(z)+\ga \WV_x+c\WV\\
&=-\WV_t-\sup\limits_{H\in \aH}\left(\frac{ \WV_{xx}}{2}\int_0^\infty H(z)^2 \d F(z)+\WV_x\int_0^\infty H(z) \d F(z) \right)+\ga \WV_x+c\WV=-\WV_t-{\cal L}\WV.
\end{align*}
Because $R^*_{s-}\leq d^{*}(T-s)$ for $s>t$ and noticing \eqref{wvx}, we conclude that the second expectation in \eqref{itofor} is also zero.

Now we are left with
\begin{align*}
\WV(x,t)=\E \Bigg[ \int_t^{T\wedge \deathtime^*}e^{-c(s-t)}\WV_x(R^*_{s-},s)\d L^{*c}_s \Bigg]-\E \Bigg[\sum\limits_{t\leq s\leq T\wedge\deathtime^*}e^{-c(s-t)}(\WV(R^*_s,s)-\WV(R^*_{s-},s))\Bigg].
\end{align*}
Notice $\d L^{*c}_s=0$ unless $R^*_{s-}=d^{*}(T-s)$, so
\[\WV_x(R^*_{s-},s)\d L^{*c}_s=\WV_x(d^{*}(T-s),s)\d L^{*c}_s=\d L^{*c}_s. \]
Recall that $R^*_s$ is continuous for $s>t$, so
\begin{align*}
\E \Bigg[\sum\limits_{t\leq s\leq T\wedge\deathtime^*}e^{-c(s-t)}(\WV(R^*_s,s)-\WV(R^*_{s-},s))\Bigg]
=\E\Big[\WV(R^*_t,t)-\WV(R^*_{t-},t)\Big].
\end{align*}
Putting the above three equations together, we have
\begin{align*}
\WV(x,t) &=\E \Bigg[ \int_t^{T\wedge \deathtime^*}e^{-c(s-t)} \d L^{*c}_s \Bigg]- \E\Big[\WV(R^*_t,t)-\WV(R^*_{t-},t)\Big]\\
&=\E \Bigg[ \int_{t-}^{T\wedge \deathtime^*}e^{-c(s-t)} \d L_s^* \Bigg]-\E\Big[(L_t^*- L^*_{t-})+\WV(R^*_t,t)-\WV(R^*_{t-},t)\Big].
\end{align*}
If $R^{*}_{t-}\leq d^{*}(T-t)$, then $L^*_{t}-L^*_{t-}=0$ and $R^*_t=R^*_{t-}$, so
\[(L_t^*- L^*_{t-})+\WV(R^*_t,t)-\WV(R^*_{t-},t)=0\]
If $R^{*}_{t-}> d^{*}(T-t)$, then $R^*_t=d^{*}(T-t)$. Because $\WV_x(y,t)=1$ for $y\geq d^{*}(T-t)$, we also obtain
\[(L_t^*- L^*_{t-})+\WV(R^*_t,t)-\WV(R^*_{t-},t)=(L_t^*- L^*_{t-})+R^{*}_{t}-R^{*}_{t-}=0.\]
Now we conclude that
\begin{align*}
\WV(x,t) &=\E \Bigg[ \int_{t-}^{T\wedge \deathtime^*}e^{-c(s-t)} \d L^{*}_s \Bigg].
\end{align*}
The right hand side is by definition no more than $V(x,t)$. This completes the proof of \thmref{thm:averi}.

\section{Proof of \lemref{prop:vR}.}\label{sec:proof2}
\setcounter{equation}{0}

By \eqref{v_pb00},
\be\label{vND}
v_t-A\left(-\frac{v_{xx}}{v_{x}}\right) v_{xx}-B\left(-\frac{v_{xx}}{v_{x}}\right)v_x+\ga v_x+cv=0\quad \hbox{in}\quad \ND.
\ee
Differentiating \eqref{vND} w.r.t. $x$ and $t$, respectively, using \eqref{Ap} we obtain
\be\label{vxND}
v_{tx}-\left[A\left(-\frac{v_{xx}}{v_{x}}\right) v_{xxx}+B\left(-\frac{v_{xx}}{v_{x}}\right)v_{xx} \right]+\ga v_{xx}+cv_x=0\quad \hbox{in}\quad \ND,
\ee
and
\be\label{vtND}
v_{tt}-\left[A\left(-\frac{v_{xx}}{v_{x}}\right) v_{xxt}+B\left(-\frac{v_{xx}}{v_{x}}\right)v_{xt} \right]+\ga v_{xt}+cv_t=0\quad \hbox{in}\quad \ND.
\ee
Since
\[
0<A(\la)\leq A\left(-\frac{v_{xx}}{v_{x}}\right)\leq \mu_2,\quad B\left(-\frac{v_{xx}}{v_{x}}\right)\leq \mu_2
\]
and
\[
A\left(-\frac{v_{xx}}{v_{x}}\right),\;B\left(-\frac{v_{xx}}{v_{x}}\right)\in C^{\al,\al/2}(\Q_T)
\]
(owing to that $v_x$, $v_{xx}\in C^{\al,\al/2}(\Q_T)$ and Lemma \ref{G}), we can apply the Schauder estimate (see \cite{Li96} Theorem 4.23) to \eqref{vxND} and \eqref{vtND}, respectively, to obtain
\[
v_x,\; v_t\in C^{2+\al,1+\al/2}(\ND).
\]
Suppose $v_{xx}(x_{0},t_{0})=0$ for some $(x_{0},t_{0})\in\ND$. Because $v_{xx}\leq 0$, $(x_{0},t_{0})$ is a maximizer point for $v_{xx}$. Hence the first order condition gives $v_{xxx}(x_{0},t_{0})=0$. By \eqref{vxND},
\[v_{tx}(x_{0},t_{0})+cv_{x}(x_{0},t_{0})=0,\]
which contradicts \eqref{vx} and \eqref{vxt}. Therefore $v_{xx}<0$ in $\ND$.

%%%%%%%%%%%%%%%%%%%%%%%%%%%%%%%%%%%%%%%%%%%%%%%%%%%%%%%%%%%%%%%%%

\section{Proof of \lemref{prop:ax}.}\label{sec:proof3}
\setcounter{equation}{0}

Equation \eqref{A0} is derived from the boundary condition in \eqref{u_pb}. And \eqref{A} is an immediate consequence of \eqref{A0} and \eqref{hx}. So it is only left to prove \eqref{hx}. In $\ND$, the equation in \eqref{u_pb} holds, i.e.,
\bee
v_{xt}-A\left(-\frac{v_{xx}}{v_x}\right) v_{xxx}-B\left(-\frac{v_{xx}}{v_x}\right)v_{xx}+\ga v_{xx}+cv_x=0.
\eee
Because
\[
v_{xxx}=\left[\p_x\left(-\frac{v_{x}}{v_{xx}}\right)+1\right]\frac{v_{xx}^2}{v_x},
\]
it follows
\bee
v_{xt}-A\left(-\frac{v_{xx}}{v_x}\right) \left[\p_x\left(-\frac{v_{x}}{v_{xx}}\right)+1\right]\frac{v_{xx}^2}{v_x}-B\left(-\frac{v_{xx}}{v_x}\right)v_{xx}+\ga v_{xx}+cv_x=0.
\eee
Dividing $v_{x}$ yields
\bee
-A\left(-\frac{v_{xx}}{v_x}\right) \left[\p_x\left(-\frac{v_{x}}{v_{xx}}\right)+1\right]\frac{v_{xx}^2}{v_x^{2}}-B\left(-\frac{v_{xx}}{v_x}\right)\frac{v_{xx}}{v_x}+\ga \frac{v_{xx}}{v_x}+c=-\frac{v_{xt}}{v_x}\leq 0,
\eee
by \eqref{vx} and \eqref{vxt}.
Denote $\nu=-v_x/v_{xx}$. Then the above inequality reads
\bee
\frac{A\left(\frac{1}{\nu}\right)(\nu_x+1)}{2\nu^2}-\frac{B\left(\frac{1}{\nu}\right)}{\nu}+\frac{\ga}{\nu}-c\geq0,
\eee
so
\be\label{nux}
\nu_x\geq\frac{1}{A(1/\nu)}\left[\left(-\frac{1}{\nu}A\left(\frac{1}{\nu}\right)+B\left(\frac{1}{\nu}\right)-\ga\right)\nu+c\nu^2\right]
=\frac{1}{A( 1/\nu)}\left[f\left(\frac{1}{\nu}\right)\nu+c\nu^2\right].
\ee
By \eqref{vga}, $\nu\geq 1/\la>0$. By \eqref{fp}, $f$ is decreasing, so $f(1/\nu)\geq f(\la)=0$. Together with \eqref{A_b} we obtain $\nu_x\geq 2c$.
This completes the proof.

\end{appendices}

\newpage

\bibliographystyle{plainnat}

\end{document}